\documentclass[conference]{IEEEtran}
\usepackage[available,functional]{ndssbadges}
\usepackage{url}            
\usepackage{booktabs}       
\usepackage{amsfonts}       
\usepackage{bbm}
\usepackage{nicefrac}       
\usepackage{graphicx}
\usepackage{amsmath}
\usepackage{makecell}
\usepackage{listings}
\usepackage{wrapfig}

\usepackage{lipsum,xcolor}
\usepackage{setspace}
\usepackage{booktabs, tabularx, colortbl}
\usepackage{subfigure}
\usepackage{array, xspace, url, amsmath, bm, array,amsfonts, balance,multirow, color, kantlipsum, float, multicol, ragged2e, soul,threeparttable}
\usepackage{amsthm}
\usepackage{caption}
\usepackage{adjustbox}
\usepackage{tcolorbox}

\let\labelindent\relax
\usepackage{enumitem}
\usepackage{hyperref}

\usepackage[linesnumbered, ruled]{algorithm2e} 
\usepackage{algorithmicx}

\usepackage{tikz}

\newcommand{\wrt}{\hbox{{w.r.t.}}\xspace}

\newcommand{\first}{\textsf{(i)}\xspace}
\newcommand{\second}{\textsf{(ii)}\xspace}
\newcommand{\third}{\textsf{(iii)}\xspace}
\newcommand{\fourth}{\textsf{(iv)}\xspace}

\newcommand{\eg}{\textsf{\emph{e.g.,}}\xspace}
\newcommand{\ie}{\textsf{\emph{i.e.,}}\xspace}
\newcommand{\downer}{\textsf{DO}\xspace}
\newcommand{\downers}{\textsf{DOes}\xspace}
\newcommand{\mowner}{\textsf{MO}\xspace}

\newcommand{\fixme}[1]{\textcolor{red}{#1}}
\newcommand{\fixmeb}[1]{\textcolor{blue}{#1}}
\newcommand{\tabincell}[2]{\begin{tabular}{@{}#1@{}}#2\end{tabular}}

\newcommand{\integers}{\mathbb{Z}}
\newcommand{\ring}{\mathcal{R}}

\newcommand{\conv}{\mathsf{Conv2d}}
\newcommand{\grad}[1]{\nabla_{#1}}
\newcommand{\relu}{\mathsf{ReLU}}
\newcommand{\drelu}{\mathsf{DReLU}}

\renewcommand{\vector}[1]{\mathbf{#1}}

\newcommand{\enc}[1]{[\![#1]\!]}
\newcommand{\pk}{\mathsf{pk}}
\newcommand{\sk}{\mathsf{sk}}
\newcommand{\angleshare}[1]{\langle #1 \rangle}
\newcommand{\share}[1]{\left\langle #1 \right\rangle}

\newcommand{\simulator}{\mathsf{Sim}}
\newcommand{\hybrid}{\mathsf{Hyb}}

\newcommand{\paraspace}{\vspace{0.01in}}
\newcommand{\parab}[1]{\paraspace\noindent{\bf #1.}}

\definecolor{shadecolor}{RGB}{255,213,169}

\newcommand{\sys}{\textsf{Pencil}\xspace}
\newcommand{\sysOpt}{\textsf{Pencil}${}^+$\xspace}
\newcommand{\pcml}{\textsf{Pencil}\xspace}

\newcommand{\productProtocol}{\mathsf{P}}
\newcommand{\productProtocolPrep}{\productProtocol_{\mathsf{Prep}}}
\newcommand{\productProtocolOnline}{\productProtocol_{\mathsf{Online}}}
\newcommand{\reverse}{\mathsf{rev}}

\newtheoremstyle{nospacetheorem}
  {0.2ex}
  {0.2ex}
  {\slshape}
  {\parindent}
  {\scshape}
  {.}
  {0.5em}
  {}

\newtheorem{definition}{Definition}[section]
\newtheorem{theorem}{Theorem}[section]
\newtheorem{corollary}{Corollary}[section]

\makeatletter
\renewenvironment{proof}[1][\proofname]{\par
  \vspace{0pt}
  \pushQED{\qed}%
  \normalfont
  \topsep0pt \partopsep0pt 
  \trivlist
  \item[\hskip\labelsep\indent
        \scshape
    #1\@addpunct{.}]\ignorespaces
}{%
  \popQED\endtrivlist\@endpefalse
}
\makeatother

\begin{document}
%
\title{\sys: Private and Extensible Collaborative Learning without the Non-Colluding Assumption}

\author{\IEEEauthorblockN{
Xuanqi Liu\IEEEauthorrefmark{1},
Zhuotao Liu\thanks{Zhuotao Liu is the corresponding author.}\IEEEauthorrefmark{1}\IEEEauthorrefmark{2},
Qi Li\IEEEauthorrefmark{1}\IEEEauthorrefmark{2},
Ke Xu\IEEEauthorrefmark{1}\IEEEauthorrefmark{2},
and
Mingwei Xu\IEEEauthorrefmark{1}\IEEEauthorrefmark{2}
}
\IEEEauthorblockA{
\IEEEauthorrefmark{1}Tsinghua University,
\IEEEauthorrefmark{2}Zhongguancun Laboratory \\
lxq22@mails.tsinghua.edu.cn, 
\{zhuotaoliu, qli01, xuke, xumw\}@tsinghua.edu.cn
}
}


%


\IEEEoverridecommandlockouts
\makeatletter\def\@IEEEpubidpullup{6.5\baselineskip}\makeatother
\IEEEpubid{\parbox{\columnwidth}{
    Network and Distributed System Security (NDSS) Symposium 2024\\
    26 February - 1 March 2024, San Diego, CA, USA\\
    ISBN 1-891562-93-2\\
    https://dx.doi.org/10.14722/ndss.2024.24512\\
    www.ndss-symposium.org
}
\hspace{\columnsep}\makebox[\columnwidth]{}}

\maketitle

\begin{abstract}
    The escalating focus on data privacy poses significant challenges for collaborative neural network training, where data ownership and model training/deployment responsibilities reside with distinct entities. 
    Our community has made substantial contributions to addressing this challenge, proposing various approaches such as federated learning (FL) and privacy-preserving machine learning based on cryptographic constructs like homomorphic encryption (HE) and secure multiparty computation (MPC). 
    However, FL completely overlooks model privacy, and HE has limited extensibility (confined to only one data provider). While the state-of-the-art MPC frameworks provide reasonable throughput and simultaneously ensure model/data privacy, they rely on a critical non-colluding assumption on the computing servers, and relaxing this assumption is still an open problem. 

    In this paper, we present \sys, the first private training framework for collaborative learning that simultaneously offers data privacy, model privacy, and extensibility to multiple data providers, without relying on the non-colluding assumption.  
    Our fundamental design principle is to construct the $n$-party collaborative training protocol based on an efficient two-party protocol, and meanwhile ensuring that switching to different data providers during model training introduces no extra cost. 
    We introduce several novel cryptographic protocols to realize this design principle and conduct a rigorous security and privacy analysis. 
    Our comprehensive evaluations of \sys demonstrate that \first models trained in plaintext and models trained privately using \sys exhibit nearly identical test accuracies; \second The training overhead of \sys is greatly reduced: 
    \sys achieves $10 \sim 260\times$ higher throughput and 2 orders of magnitude less communication than prior art; \third \sys is resilient against both existing and adaptive (white-box) attacks.
\end{abstract}


%
\IEEEpeerreviewmaketitle

\section{Introduction}


Recent years witnessed significant development and applications of machine learning, in which the most successful ones were the use of deep neural networks. 
Effective training of neural network models hinges on the availability of a substantial corpus of high-quality training data. 
However, in real-world business scenarios, the entities possessing data (\ie data providers, \downers) and the entity seeking to utilize data for training and deployment of a machine learning model (\ie a model deployer, \mowner) are distinct parties. One representative example of this is the collaborative anti-money laundering (AML)~\cite{collaboration-is-key-of-AML, enhance-collaboration-of-aml}, where multiple financial organizations (\eg banks) are interested in training anti-money laundering models based on the cell phone records (more accurately, a range of features engineered from these records) that are owned by the cellular service providers. 
However, the growing emphasis on data privacy, accompanied by the emergence of stringent regulations and laws such as the General Data Protection Regulation (GDPR)~\cite{GDPR}, has rendered the direct sharing of raw data across multiple organizations infeasible or even unlawful. Consequently, accessing training data, particularly private-domain data not readily accessible on the public Internet, presents a significant research challenge.



Concretely, the above form of collaborative learning raises a combination of three critical requirements. \first Data Privacy: the training data of different \downers should be kept confidential. This is the preliminary requirement in privacy-preserving machine learning. \second Model Privacy: the trained model is only revealed to the \mowner (\ie the \mowner can perform model inference independently), but not disclosed to the \downers. In real-world business scenarios, the post-training model is proprietary. Thus, model privacy is equally important as data privacy. \third Extensibility: because the \mowner often seeks to train the model with multiple (and heterogeneous) data providers, it is critical to ensure that the \mowner can incrementally collaborate with different \downers. Concretely, we define the this form of learning paradigm \emph{private and extensible collaborative learning}. 

Over the past few years, our community has proposed significant research in this regard. However, none of the existing frameworks meet our aforementioned requirements. 
\first Federated Learning (FL) is a widely applied collaborative learning scheme that enables many participants to collectively train a model without sharing their original data. However, the model in FL is synchronized with every participant in every round of the global model update (\eg \cite{McMahan2016Communication, Bonawitz2017Practical, Bonawitz2019Towards, Konecn2016Federated, Gupta2018Distributed, Gao2020End, Jeon2020Privacy, li2023martfl}). Thus, in FL, the model privacy is simply ignored.
\second The second stream of approaches relies on secure multiparty computation (MPC) framework, where the \downers upload their data as  secret shares to several third-party computing servers (\eg \cite{Nandakumar2019towards, Tian2022sphinx, Hesamifard2018Privacy, Sav2021Poseidon, Mohassel2017secureml, Agrawal2019quotient, Mohassel2018aby3, Wagh2019securenn, Patra2022blaze, Chaudhari2020trident, Koti2021Tetrad}). 
The privacy and security guarantees of these proposals require a critical \emph{non-colluding assumption} that the number of colluding servers must be below a threshold. In fact, in most of these proposals, two colluding servers are sufficient to break the protocol. 
Thus, the recent SoK on cryptographic neural-network computation~\cite{Ng2023sok} considers that ``realizing high-throughput and accurate private training without non-colluding assumptions'' the first open problem. 
A strawman design to relax the non-colluding assumption is adopting the $n$-out-of-$n$ sharing schemes~\cite{Damgrd2019NewPF, Cramer2018SPDZ2k} in which all the \downers and the \mowner participate as the computing servers. Yet, the training efficiency is significantly limited (see quantitative results in \S~\ref{subsec:evaluation:comparison-prior-art}).

\subsection{Our Contributions}

In this paper, we present \sys\footnote{\sys: \textbf{P}rivate and \textbf{e}xte\textbf{n}sible \textbf{c}ollaborative mach\textbf{i}ne \textbf{l}earning.}, a novel system that meets all three aforementioned requirements of private and extensible collaborative learning, \emph{without the non-colluding assumption}. 
At its core, \sys reduces the multi-party  collaborative learning scenario into the 2-party server/client computation paradigm: at each training step, the \mowner could choose any one of the \downers to collaborate with, and switching between \downers in different training steps introduces no extra cost. 
This learning protocol is fundamentally different from the existing MPC-based designs that exhibit a tradeoff between non-collusion and extensibility. 
Specifically, if the \mowner and \downers choose to secretly share their model and data to third-party MPC servers, extensibility is achieved as the data from multiple \downers can simultaneously contribute to the training, yet they require non-collusion among these third-party MPC servers (\eg \cite{Mohassel2017secureml, Agrawal2019quotient, Mohassel2018aby3, Wagh2019securenn, Patra2022blaze, Chaudhari2020trident, Koti2021Tetrad}). On the contrary, the \mowner and \downers can avoid the non-collusion assumption by \first acting as the computing servers themselves and \second adopting the $n$-out-of-$n$ secret sharing schemes~\cite{Damgrd2019NewPF, Cramer2018SPDZ2k}, which, however, sacrifices extensibility because including additional \downers will significantly increase the time and communication overhead. 
\sys eliminates the tradeoff between non-collusion and extensibility by constructing a secure multi-\downers-single-\mowner training process using our 2-party server/client training protocol (thus relaxing the non-collusion assumption), and meanwhile enabling the \mowner to securely obtain model weights after each training step so that it can switch to an arbitrary \downer in the next training step (thus realizing extensibility).

At a very high level, our 2-party training protocol tactically combines the primitives of the 2-party secure computation protocols and Homomorphic Encryption (HE) cryptosystem for efficient non-linear and linear computation, respectively. 
Specifically, based on the recent development of efficient private inference~\cite{Huang2022cheetah, Hao2022iron}, we achieve efficient private training by constructing novel cryptography protocols to support backpropagation (such as efficient computation of the product of two secret-shared tensors in computing weight gradient). 
To further improve training efficiency, we propose a novel preprocessing design to pre-compute heavy HE-related operations (\eg matrix multiplications and convolutions) offline. This preprocessing design is fundamentally different from the preprocessing technique in \cite{Mishra2020delphi} that, given a fixed model, accelerates the private inference for a single input sample. 
Finally, as a contribution independent of our training protocols, we implement the key operators in the BFV HE cryptosystem (\cite{Fan2012bfv, Bajard2017RNSBFV}) on GPU. By fully exploiting the parallelism in the math construction, our open-sourced prototype achieves, on average, $10\times$ speedup compared with the CPU based implementation of SEAL~\cite{SEAL}. 

In summary, our contributions are:

\begin{itemize}[leftmargin=*, parsep=0.2ex, topsep=0.2ex, partopsep=0.2ex]
    \item We develop \sys, an efficient training framework to realize efficient collaborative training for the single-\mowner-multi-\downers scenario. \sys protects data privacy against the \mowner and model privacy against the \downers, without requiring the non-colluding assumption. \sys supports both training from scratch and fine-tuning on existing models. 
    \item We design a set of novel  protocols in \sys, and provide rigorous analysis to prove their security/privacy guarantees. Independent of our protocols, we implement a full GPU version of the BFV HE cryptosystem that is usable by any HE-based machine learning applications.  
    \item We extensively evaluate \sys. Our results show that: \first models trained in \sys and trained in plaintext have nearly identical test accuracies; 
    \second With our protocol and hardware optimizations, the training (from scratch) of simple convolutional networks in \sys converges within 5 hours for the MNIST task. With transfer learning, a complex classifier (\eg ResNet-50 with 23 million parameters) could be fine-tuned in \sys within 8 hours for the CIFAR10 task. These results show non-trivial performance improvements over closely related art;
    \third \sys is extensible to multiple \downers with no overhead increase. In the case where individual \downers possess heterogeneous or even biased datasets, \sys demonstrates significant performance gains by incorporating more \downers;   
    \fourth we also experimentally show that \sys is robust against existing attacks and adaptive (white-box) attacks designed specifically for our protocols.
\end{itemize}

\subsection{Related Work}
\label{subsec:intro:related-work}

The increasingly growing data collaboration among different parties sparked significant research in privately training  machine learning models. We divide prior art into three subcategories, as summarized in Table~\ref{table:summary-related-works}.

Federated learning (FL) is the pioneering machine learning scheme that treats training data privacy as the first-class citizen. In (horizontal) FL (\eg \cite{McMahan2016Communication, Bonawitz2017Practical, Bonawitz2019Towards, Konecn2016Federated, Gupta2018Distributed, Gao2020End, li2023martfl}), each client (\ie \downer) trains a local model on its private data, and then submits the obfuscated local models to the model aggregator (\ie \mowner). 
Due to the model synchronization in each training iteration, the privacy of the global model is overlooked by FL. In particular, FL disproportionately benefits \downers because each \downer learns the global model even if it only partially contributes to the training. As a result, in the case where the \mowner has an (exclusive) proprietary interest in the final model, FL is not the turnkey solution. 

To simultaneously protect data and model privacy, the community proposed a branch of distributed training approaches based on secure multiparty computation (MPC, \eg \cite{Mohassel2017secureml, Agrawal2019quotient, Mohassel2018aby3, Wagh2019securenn, Patra2022blaze, Chaudhari2020trident, Koti2021Tetrad, Dalskov2021Fantastic}). The common setting studied in this line of art is a server-assisted model where a group of \downers upload secret-shared version of their private data to (at least two) \emph{non-colluding} third-party computing servers to collectively train a model on the secret-shared data. The model is also secret-shared among all servers during the training process, perfectly protecting model privacy. 
As explained before, this type of approaches suffers from a critical tradeoff between non-collusion and extensibility.

The third subcategory of research depends on homomorphic encryption (HE). These works (\eg \cite{Nandakumar2019towards, Hesamifard2018Privacy, Sav2021Poseidon, Tian2022sphinx}) typically focus on outsourced training where a \downer employs a cloud service to train a model by uploading homomorphically-encrypted data, and the final model trained on the cloud is also encrypted. Afterward, the trained model is handed over to the \downer or deployed on the cloud as an online inference service usable by the \downer. This paradigm, however, is ill-suited in the single-\mowner-multi-\downers collaborative learning scenario because the plaintext model is not available to the \mowner for deployment, let alone being further fine-tuned using other \downers' data (since it requires different HE public/secret key sets).

Differential privacy (DP) is a general technique that can be combined with the aforementioned techniques in private learning. 
For instance, DPSGD~\cite{Abadi2016dpsgd} could be used in FL to protect the privacy of individual data points. To defend against gradient matching attacks~\cite{Zhao2020idlg}, Tian et al. \cite{Tian2022sphinx} combines the HE approach with DP mechanics to randomize the weight updates against the cloud service provider. As shown in \S~\ref{subsec:design:linear:backward}, \sys allows the \downers to (optionally) add perturbations to the weight updates to achieve additional differential-privacy guarantees on their training data.

\begin{table*}
    \centering \small
    \begin{threeparttable}[b]\begin{tabular}{cc|c|cccc}
        \hline
        Category & \tabincell{c}{Representative\\framework} & Techniques used\tnote{*} & \tabincell{c}{Data\\privacy} & \tabincell{c}{Model\\privacy} & \tabincell{c}{Against\\collusion} & Extensibility \\

        \hline
        Horizontal FL & \cite{McMahan2016Communication, Bonawitz2017Practical, Bonawitz2019Towards} & Local SGD & $\checkmark$ & $\times$ & $\checkmark$ & $\checkmark$ \\
        Vertical FL & \cite{Gupta2018Distributed, Gao2020End, Jeon2020Privacy} & Local SGD & $\checkmark$ & $\times$ & $\checkmark$ & $\checkmark$ \\

        \hline
        MPC (2 servers) & \cite{Agrawal2019quotient, Mohassel2017secureml} & GC, SS & $\checkmark$ & $\checkmark$ & $\times$ & $\checkmark$\tnote{$\dagger$} \\
        MPC (3 servers) & \cite{Mohassel2018aby3, Patra2022blaze, Wagh2019securenn} & GC, SS & $\checkmark$ & $\checkmark$ & $\times$ & $\checkmark$\tnote{$\dagger$} \\
        MPC (4 servers) & \cite{Chaudhari2020trident, Koti2021Tetrad, Dalskov2021Fantastic} & GC, SS & $\checkmark$ & $\checkmark$ & $\times$ & $\checkmark$\tnote{$\dagger$} \\
        MPC ($n$ servers) & \cite{Damgrd2019NewPF, Cramer2018SPDZ2k} & GC, SS & $\checkmark$ & $\checkmark$ & $\checkmark$ & $\checkmark$\tnote{$\ddagger$} \\ 

        \hline
        Data outsourcing / cloud & \cite{Nandakumar2019towards, Hesamifard2018Privacy} & HE & $\checkmark$ & $\times$ & N/A & $\times$ \\
        Data outsourcing / cloud & \cite{Tian2022sphinx} & HE, DP & $\checkmark$ & $\times$ & N/A & $\times$ \\

        \hline
        \pcml & Ours & HE, SS, DP & $\checkmark$ & $\checkmark$ & $\checkmark$ & $\checkmark$ \\
        \hline
    \end{tabular}
    \begin{tablenotes}
        \footnotesize
        \item[*] SGD is for stochastic gradient descent, GC for garbled circuits, SS for secret sharing, HE for homomorphic encryption and DP for differential privacy.
        \item[$\dagger$] {
        If \mowner and \downers choose to secretly share their model and data to third-party MPC servers, extensibility is achieved but the approaches are secure only if these servers are not colluding with each other. 
        }
        \item[$\ddagger$] {The general $n$-PC protocol against collusion suffers from a scalability problem: including more parties would greatly increase the computation overhead. See \S~\ref{subsec:evaluation:comparison-prior-art} for experimental results.}
    \end{tablenotes}
    \caption{Comparison of prior art related with private collaborative training.}
    \label{table:summary-related-works}
\end{threeparttable}\end{table*}

\subsection{Assumptions and Threat Model}
We consider the scenario where one \mowner and multiple \downers participate in a collaborative machine learning system. All parties are \emph{semi-honest}, and any of them \emph{may collude} to infer the private information (model parameters of \mowner or the data owned by the \downers) of other parties. The architecture of the learning model is known by all parties. 
As described in the previous sections, we reduce this $n$-party setting to the $2$-party paradigm where only the \mowner and one of the \downers interact at one single training step. By establishing our privacy guarantees over the two-party protocols, we obtain the security for the general $n$-party extensible machine learning scheme \emph{without the non-colluding assumption}.

During training, a \downer obtains the model prediction result of its own training dataset. Like previous works~\cite{Sav2021Poseidon, Huang2022cheetah, Mohassel2018aby3, Koti2021Tetrad}, we do not consider inference attacks based solely on prediction results (\eg \cite{Tramer2016Stealing, Shokri2017Membership}). 

\section{Preliminaries}

\subsection{Notations}

Vectors, matrices and tensors are denoted by boldfaced Latin letters (\eg $\vector{W}, \vector{x})$, while polynomials and scalars are denoted by italic Latin letters (\eg $W, x$). All the gradients are relative to the loss function $L$ of the deep learning task, \ie $\grad{W} = \grad{W} L = \partial{L}/\partial{W}$. $[B]$ denotes the integer interval $[0, B) \bigcap \mathbb{Z}$. Symbol $\vector{x}' \sim \vector{x}$ means that tensor $\vector{x}'$ is of the same shape as $\vector{x}$. If not explicitly specified, a function (\eg $\relu$) applied to a tensor is the function applied to all of its elements.

\subsection{Lattice-based Homomorphic Encryption}

\sys uses the BFV leveled homomorphic encryption cryptosystem based on the RLWE problem with residual number system (RNS) optimization~\cite{Fan2012bfv,Bajard2017RNSBFV}. In detail, the BFV scheme is constructed with a set of parameters $\{N, t, q\}$ such that the polynomial degree $N$ is a power of two, and $t$, $q$ represent plaintext and ciphertext modulus, respectively. The plaintext space is the polynomial ring $\ring_{t, N} = \mathbb{Z}_t[X]/(X^N+1)$ and the ciphertext space is $\ring_{q,N}^2$. Homomorphism is established on the plaintext space $\ring_{t, N}$, supporting addition and multiplication of polynomials in the encrypted domain. We denote the homomorphically encrypted ciphertext of polynomial $x$ as $\enc{x}$. For a tensor $\vector{x}$, encryption $\enc{\vector{x}}$ requires an encoding method to first convey the tensor into the polynomial ring, which we cover in Appendix~\ref{appendix:polynomial-encoding}.

\subsection{Additive Secret Sharing and Fixed-point Representation}\label{subsec:prelim:additive-secret-sharing}

We utilize the additive secret-sharing scheme upon the ring $\mathbb{Z}_t$ (integers modulo $t$) with $t=2^\ell$. 
If an integer $x \in \integers_t$ is shared between a pair of \mowner and \downer, then \mowner (Party 0) has $\share{x}_0$ and \downer (Party 1) has $\share{x}_1$ such that $x = \share{x}_0 + \share{x}_1$. For simplicity, $\share{x}$ denotes $x$ is shared between the two parties.

Machine learning typically involves decimal numbers rather than integers. To adapt to the BFV scheme and integer-based secret sharing, we use a fixed-point representation of decimal numbers.
A decimal $\tilde{x} \in \mathbb{R}$ is represented as an integer $x = \mathsf{Encode}(\tilde{x}) = \lfloor \tilde{x} \cdot 2^f \rfloor \in \mathbb{Z}$, with a precision of $f$ bits. After every multiplication, the precision inflates to $2f$, and a truncation is required to keep the original precision. Since we use $\mathbb{Z}_t$ rather than $\mathbb{Z}$, we require all intermediate results in their decimal form $\tilde{x} \in \mathbb{R}$ not to exceed $\pm t/2^{2f+1}$, to prevent overflow. In the rest of the paper, unless stated otherwise, all scalars and elements of tensors are in $\mathbb{Z}_t$.

\subsection{Neural Network Training}

A typical neural network (NN) consists of several \emph{layers}, denoted by a series of functions $f_i(\vector{x}), i\in {1, 2, \cdots, \ell}$, and for an input sample $\vector{x}$, the neural network's output is $\vector{y} = f_\ell(f_{\ell-1}(\cdots f_1(\vector{x})\cdots))$, \ie feeding the input through each layer. It is easy to extend the above modeling to neural networks with branches. The NN layers are roughly divided into two categories: linear layers and non-linear layers. Typically, linear layers (\eg fully connected layers and convolution layers) contain trainable parameters that could be \emph{learned} from the input samples and their corresponding labels. Given a linear layer $f_i(\vector{x}; \vector{W}_i, \vector{b}_i)$ with trainable parameters weight $\vector{W}_i$ and bias $\vector{b}_i$, we denote it as a function $f_i(\vector{x})$ when there is no confusion. 

\parab{Forward propagation} Feeding input through each layer is called \emph{forward propagation}. In forward propagation, linear layers could be abstracted as
\begin{equation}
    \label{eq:linear-layer-abstraction}
    \vector{y} = f(\vector{x}) = f(\vector{x}; \vector{W}, \vector{b}) = \vector{W}\circ \vector{x} + \vector{b},
\end{equation}
where $\circ$ is a \emph{linear operator} satisfying the following constraint
\begin{equation}
    (\vector{u}_0 + \vector{u}_1) \circ (\vector{v}_0 + \vector{v}_1) = \vector{u}_0 \circ \vector{v}_0 + \vector{u}_1 \circ \vector{v}_0 + \vector{u}_0 \circ \vector{v}_1 + \vector{u}_1 \circ \vector{v}_1.
\end{equation}
For example, in fully connected layers, $\circ$ represents the matrix-vector multiplication:
\begin{equation*}
    \vector{x} \in \integers_t^{n_i}, \vector{W} \in \integers_t^{n_o \times n_i}, \vector{b} \in \integers_t^{n_o}, \vector{y} = \vector{W} \circ \vector{x} + \vector{b} = \vector{W} \cdot \vector{x} + \vector{b};
\end{equation*}
in 2-dimensional convolution layers, $\circ$ is the convolution operation:
\begin{equation*}\begin{array}{c}
     \vector{x} \in \integers_t^{c_i\times h \times w}, \vector{W} \in \integers_t^{c_o \times c_i \times s\times s} \\
     \vector{W} \circ \vector{x} = \conv(\vector{x}; \vector{W}) \in \integers_t^{c_o \times (h-s+1) \times (w-s+1) } \\
\end{array}\end{equation*}

By contrast, the output of non-linear layers (\eg ReLU function and max pooling layers) is completely determined by the input. 

\parab{Backpropagation}
The goal of training a neural network is to find a set of trained parameters such that some loss function $L(y; t)$ calculated on network output $y$ and ground truth $t$ (labels) is minimized across a dataset $\mathcal{D} = \{(\vector{x}_i, t_i)\}_i$. In practice, one can draw a batch of $B$ input samples $(\vector{X}, \vector{t}) = \{(\vector{x}_i, t_i)\}_{i\in[B]}$ and calculate the mean loss $L = L(\vector{Y}, \vector{t}) = \frac{1}{B}L(\vector{y}_i; t_i)$. Applying the chain rule, one can reversely compute the partial derivative of each trainable weight with respect to $L$, \ie $\nabla_{\vector{W}} = \partial L / \partial \vector{W}, \nabla_{\vector{b}} = \partial L / \partial \vector{b}$ and update the weights accordingly. Since the gradients are propagated reversely from the last layer to the first layer, this process is called \emph{backpropagation}.

\section{Training in \pcml}
In this section, we introduce our core designs to realize private-preserving training of neural networks in \pcml. We begin by introducing the high-level procedure for our 2-party training protocol. Next, we explain the detailed protocols for training  linear and non-linear layers. 
Then, we introduce an optimization method that offloads the computationally heavy  operations to an offline phase. Finally, we demonstrate how to extend our methods for 2-PC to the scenario of arbitrary number of \downers.

\subsection{\pcml Training Overview} 
In the training phase of \pcml, the \downer holds all training data including labels, while the \mowner holds all trainable parameters. The network architecture is known to both parties. Except for the final output $y$ which is revealed to the \downer, all intermediate outputs of each layer are secret-shared between \downer and \mowner.

In the forward propagation, the \downer draws a batch of input data $\vector{X} = \{\vector{x}_i\}_{i\in[B]}$ to feed into the neural network, and the corresponding labels are $\vector{t}=\{t_i\}_{i\in[B]}$. For input $\vector{X}$, denote $\vector{X}_{0} = \vector{X}$. Evaluation of each layer $f_i$ is denoted as $\vector{X}_{i} = f_i(\vector{X}_{i-1})$. At the beginning, \mowner takes $\share{\vector{X}_0}_0 = \vector{0}$ and \downer takes $\share{\vector{X}_0}_1 = \vector{X}$. This is a secret-sharing of $\share{\vector{X}_i} = \share{\vector{X}_i}_0 + \share{\vector{X}_i}_1$ for $i = 0$. We keep this invariant form of secret sharing for all layers $f_i$. Essentially, $f_i$ is a secure computation protocol operating on secret shares: $f_i$ takes shares $\share{\vector{X}_{i-1}}$ from the two parties and produces shares $\share{\vector{X}}_i$ to both parties such that
$$\share{\vector{X}_i}_0 + \share{\vector{X}_i}_1 = f_i(\share{\vector{X}_{i-1}}_0 + \share{\vector{X}_{i-1}}_1).$$
The construction of such a protocol is introduced in \S~\ref{subsec:design:linear:forward}. The \mowner reveals the final propagation output $\share{\vector{Y}}_0 = \share{\vector{X}_{\ell}}_0$ to \downer, based on which \downer reconstructs the prediction result $\vector{Y}$ to calculate the loss function $L(\vector{Y}, \vector{t})$.  

In the backpropagation, each derivative $\nabla_{\vector{X}_{i}}$ is shared:
$$\share{\grad{\vector{X}_{i-1}}}_0 + \share{\grad{\vector{X}_{i-1}}}_1 = (\share{\grad{\vector{X}_i}}_0 + \share{\grad{\vector{X}_i}}_1) \odot_x \frac{\partial f_i(\vector{X}_{i-1})}{\partial \vector{X}_{i-1}}$$
With secret shared values of $\nabla_{\vector{X}_{i}}$ and $\vector{X}_{i}$, the two parties could collaborate to produce the gradients of trainable parameters in linear layers, \ie $\grad{\vector{W}_i}, \grad{\vector{b}_i}$:
$$
    \begin{aligned}
    \grad{\vector{b}_i} & = \grad{\vector{X}_i} \odot_b \frac{\partial f_i(\vector{X}_{i-1}; \vector{W}_i, \vector{b}_i)}{\partial \vector{b}_{i}} \\
    \grad{\vector{W}_i} & = \grad{\vector{X}_i} \odot_W \frac{\partial f_i(\vector{X}_{i-1}; \vector{W}_i, \vector{b}_i)}{\partial \vector{W}_{i}} \\
    \end{aligned}
$$
These weight gradients are revealed to \mowner to update the parameters (introduced in \S~\ref{subsec:design:linear:backward}). 
Note that the linear operators $\odot_x, \odot_W, \odot_b$ could be deduced from the forward propagation formula $f(\vector{X}_{i-1}) = \vector{X}_i$, according to the chain rule of derivatives. The challenge in backpropagation is how to reveal the weight gradients $\grad{\vector{W}_i}, \grad{\vector{b}_i}$ \emph{only to the \mowner}, while protecting the privacy of both intermediate outputs ($\vector{X}_i, \grad{\vector{X}_i}$) and weights themselves ($\vector{W}_i, \vector{b}_i$).

As a feed-forward network could be decomposed into a series of layers, in the following, we discuss how to evaluate one single layer in the neural network. Since the trainable parameters of a neural network are in linear layers, we mainly focus on the training protocols of linear layers. We address the evaluations of non-linear layers in \S~\ref{subsec:design:nonlinear}.

\subsection{Linear Protocols}
\label{subsec:design:linear}

We first introduce the forward propagation and backpropagation protocols of linear layers. For forward propagation, we adopt the recent development of efficient private preference~\cite{Huang2022cheetah,Hao2022iron,liu2023llms} as a strawman design, and then extend it to support \emph{batched inference instead of just single inference}. Then we elaborate on the gradient computation protocol in the backpropagation.

\subsubsection{Forward Propagation}
\label{subsec:design:linear:forward} 

Homomorphic evaluations of linear layers are fundamental to enabling privacy-preserving machine learning. We summarize the high-level protocol in Algorithm~\ref{alg:homomorphic-linear-evaluation}. Note that because the HE ciphertexts are converted to secret-shares after each linear layer evaluation, we need only to support one HE multiplication in the BFV parameters and do not require bootstrapping.
$\circ$ is a linear operator that could be decomposed into basic arithmetic addition and multiplications (\eg matrix multiplication or convolution). 

\begin{algorithm}[t]
    \caption{Evaluation of linear layer $f$}
    \label{alg:homomorphic-linear-evaluation}
    \KwIn{The input $\share{\vector{X}}$ shared between \mowner and \downer; \mowner holds the weights $\vector{W}$ and the bias $\vector{b}$.}
    \KwOut{The output shares $\share{\vector{Y}}$ of $\vector{Y}=\vector{W}\circ \vector{X} + \vector{b}$.}
    \downer sends encrypted $\enc{\share{\vector{X}}_1}$ to \mowner; \\
    \mowner evaluates $\enc{\vector{W} \circ \vector{X}} = \vector{W} \circ (\enc{\share{\vector{X}}_1} + \share{\vector{X}}_0)$ using homomorphic plaintext-ciphertext additions and multiplications; \\
    \mowner chooses random mask $\vector{s}$ and calculates $\enc{\share{\vector{Y}}_1} = \enc{\vector{W} \circ \vector{X}} - \vector{s}$; \mowner sends $\enc{\share{\vector{Y}}_1}$ back for decryption; \\ 
    \downer outputs $\share{\vector{Y}}_1$; \mowner outputs $\share{\vector{Y}}_0 = \vector{s} + \vector{b}$.
\end{algorithm}

The substantial part of computation in Algorithm~\ref{alg:homomorphic-linear-evaluation} lies in evaluating $\vector{W} \circ \enc{\vector{X}}$. Previous art in private NN inference has developed different methods to evaluate $\vector{W}\circ\enc{\vector{X}}$ for matrix-vector multiplication and 2d-convolutions. For example, \cite{Juvekar2018gazelle} uses the SIMD support of BFV cryptosystem and designs a hybrid method for ciphertext matrix multiplication and convolution, while \cite{Huang2022cheetah} and \cite{Hao2022iron} exploit the polynomial homomorphism to efficiently evaluate the ciphertext dot product in linear layers. Note that private inference of \cite{Huang2022cheetah} is actually forward propagation of batch size $B=1$, but we adapt these primitives to \emph{batched inputs} with $B>1$. The details are deferred to Appendix~\ref{appendix:polynomial-encoding}.

\subsubsection{Backpropagation}
\label{subsec:design:linear:backward}

In backpropagation, with secret shares $\share{\grad{\vector{Y}}}$, we need to compute three types of gradients, $\grad{\vector{X}}, \grad{\vector{b}}$ and $\grad{\vector{W}}$. The gradient of the inputs $\share{\grad{\vector{X}}}$ is again secret-shared between two parties, while the gradient of the parameters $\grad{\vector{W}}, \grad{\vector{b}}$ are revealed \emph{only to} the \mowner. Based on Equation~(\ref{eq:linear-layer-abstraction}), we can deduce the formula of these three types of gradients respectively. 

\parab{Calculation of $\grad{\vector{X}}$}
For $\grad{\vector{X}} = \frac{\partial f(\vector{X}; \vector{W}, \vector{b})}{\partial \vector{X}} \odot_x \grad{\vector{Y}} = \vector{W} \odot_x \grad{\vector{Y}}$, it takes a form similar to the forward propagation procedure: the two parties input secret shares $\share{\grad{\vector{Y}}}$ and the \mowner provides weights $\vector{W}$; the output is a linear operation $\odot_x$ on $\grad{\vector{Y}}$ and $\vector{W}$, and the output is shared between two parties. Therefore, we can use a protocol similar to Algorithm~\ref{alg:homomorphic-linear-evaluation} to calculate the shares $\share{\grad{\vector{X}}}$. 

\parab{Calculation of $\grad{\vector{b}}$} $\grad{\vector{b}} = \grad{\vector{Y}} \odot_b \frac{\partial f(\vector{X}; \vector{W}, \vector{b})}{\partial \vector{b}}$ takes the form of summation across all $B$ samples.\footnote{For fully connected layers, $\grad{\vector{b}}$ is simply a summation over the batch size $B$ dimension of $\grad{\vector{Y}} \in \mathbb{R}^{B \times n_o}$. $n_o$ is the output size of the FC layer. For 2d-convolution layers, $\grad{\vector{b}}$ is a summation of $\grad{\vector{Y}} \in \mathbb{R}^{B \times c_o \times (h-s+1) \times (w-s+1)}$ over three dimensions: the batch size and output image height and width dimensions. $c_o, h, w, s$ are output channels, input image height, width, and kernel size of the 2d-convolution layer.}
Therefore, the two parties can perform  summation locally on their shares respectively, and the \downer sends its share $\share{\grad{\vector{b}}}_1$ to \mowner for reconstruction.

\parab{Calculation of $\grad{\vector{W}}$} It is more challenging to calculate $\grad{\vector{W}}$ and reveal it to the \mowner without leaking information about $\vector{X}$ or $\grad{\vector{Y}}$, since both operands of $\odot$ are in secret-shared form: (for simplicity we use $\odot$ instead of $\odot_W$):
\begin{equation}
\grad{\vector{W}} = \grad{\vector{Y}} \odot \frac{\partial f(\vector{X}; \vector{W}, \vector{b})}{\partial \vector{W}} = \grad{\vector{Y}} \odot \vector{X}
\end{equation}

A straightforward solution is to let the \downer send encrypted $\enc{\share{\vector{X}}_1}$ and $\enc{\share{\grad{\vector{Y}}}_1}$ to the \mowner. The \mowner calculates the gradient in the encrypted domain as
\begin{equation}
    \enc{\grad{\vector{W}}} = 
    (\share{\grad{\vector{Y}}}_0 + \enc{\share{\grad{\vector{Y}}}_1}) \odot (\share{\vector{X}}_0 + \enc{\share{\vector{X}}_1}).
\end{equation}
The \mowner samples random mask $\vector{s}$ and sends back perturbed $\enc{\grad{\vector{W}} - \vector{s}}$ for decryption (to protect the model update against the \downer). The \downer finally sends the decryption result and \mowner could recover plaintext $\grad{\vector{W}}$.

However, in the above procedure, we need to evaluate multiplication \emph{between two ciphertexts}, which is much more expensive than \emph{ciphertext-plaintext} multiplication. 
We use the linearity of the binary operator $\odot$ to eliminate this requirement. In particular, $\grad{\vector{W}}$ could be seen as a summation of four terms: two cross terms $\share{\grad{\vector{Y}}}_0 \odot \share{\vector{X}}_1, \share{\grad{\vector{Y}}}_1 \odot \share{\vector{X}}_0$, and two “local” terms $\share{\grad{\vector{Y}}}_0 \odot \share{\vector{X}}_0, \share{\grad{\vector{Y}}}_1 \odot \share{\vector{X}}_1$ which could be calculated locally by the two parties respectively. Therefore, the \mowner could simply evaluate the cross term 
$$\enc{\grad{\vector{W}}^{\mathsf{cross}}} = \share{\grad{\vector{Y}}}_0 \odot \enc{\share{\vector{X}}_1} + \enc{\share{\grad{\vector{Y}}}_1} \odot \share{\vector{X}}_0$$
and sends back $\enc{\grad{\vector{W}}^{\mathsf{cross}} - \vector{s}}$ for decryption. The \downer then returns the decrypted result $\widetilde{\grad{\vector{W}}} = \grad{\vector{W}}^{\mathsf{cross}} - \vector{s} + \share{\grad{\vector{Y}}}_1 \odot \share{\vector{X}}_1$, and the \mowner could recover the full $\grad{\vector{W}}$ by adding its local term $\share{\grad{\vector{Y}}}_0 \odot \share{\vector{X}}_0$ and mask $\vector{s}$. We illustrate this protocol in Algorithm~\ref{alg:weight-gradient-calculation}.

\begin{algorithm}[t]
    \caption{Weight gradient $\grad{\vector{W}}$ calculation}
    \label{alg:weight-gradient-calculation}
    \setlength{\belowdisplayskip}{0.3ex} \setlength{\belowdisplayshortskip}{0.3ex}
    \setlength{\abovedisplayskip}{0.3ex} \setlength{\abovedisplayshortskip}{0.3ex}
    \KwIn{\mowner and \downer input secret shares of $\share{\vector{X}}$ and $\share{\grad{\vector{Y}}}$.}
    \KwOut{\mowner receives $\grad{\vector{W}} = \grad{\vector{Y}} \odot \vector{X}$.}
    \downer sends encrypted $\enc{\share{\vector{X}}_1}, \enc{\share{\grad{\vector{Y}}}_1}$ to \mowner; \\
    \mowner evaluates 
    $$\enc{\grad{\vector{W}}^{\mathsf{cross}}} = \share{\grad{\vector{Y}}}_0 \odot \enc{\share{\vector{X}}_1} + \enc{\share{\grad{\vector{Y}}}_1} \odot \share{\vector{X}}_0$$ \\
    \mowner chooses random mask $\vector{s}$ and sends $\enc{\grad{\vector{W}}^{\mathsf{cross}} - \vector{s}}$ back for decryption; \\ 
    \downer evaluates
    $$\widetilde{\grad{\vector{W}}} = \grad{\vector{W}}^{\mathsf{cross}} - \vector{s} + \share{\grad{\vector{Y}}}_1 \odot \share{\vector{X}}_1$$ \\
    \downer adds a perturbation $\vector{e}$ to $\widetilde{\grad{\vector{W}}}$;\label{alg:step:weight-dp} \\
    \mowner finishes by calculating
    $$\grad{\vector{W}} = \widetilde{\grad{\vector{W}}} + \vector{s} + \share{\grad{\vector{Y}}}_0 \odot \share{\vector{X}}_0$$
\end{algorithm}

\parab{Incorporating Differential Privacy (DP) for Weight Updates}
To ensure independent model deployment, weight updates (\ie $\grad{\vector{W}}, \grad{\vector{b}}$) should be revealed to \mowner in plaintext. 
Recent art \cite{Zhu2019DLG, Zhao2020idlg, Geiping2020inverting} argues that these updates may leak private information about the training data. To address this concern, 
we integrate the DP mechanism into our framework. Specifically, we allow \downer to add perturbations to the gradients, as shown in Step~\ref{alg:step:weight-dp} of Algorithm~\ref{alg:weight-gradient-calculation}: 
\begin{equation}
    \vector{e} \leftarrow \mathcal{N}\left(0, \frac{\sigma^2 C^2}{B} \mathbb{I}\right),
\end{equation}
where $B$ is the batch size, and $C$ is the estimated the upper bound of L2 norm of the gradients \wrt a single sample. This noise term is added to $\widetilde{\grad{\vector{W}}}$ in Algorithm~\ref{alg:weight-gradient-calculation} and to the summation of \downer's own share of $\grad{\vector{b}}$. At a high level, our design is a secret-shared version of the DPSGD algorithm \cite{Abadi2016dpsgd} (by considering $C$ as the bound to clip gradients). 

\subsection{Non-linear Protocols}
\label{subsec:design:nonlinear}

To evaluate non-linear layers, \cite{Huang2022cheetah} proposes various MPC-based protocols utilizing OT extension~\cite{Kolesnikov2013ImprovedOE, Yang2020Ferret}. As \cite{Huang2022cheetah} focuses on private inference (\ie forward propagation), we extend their work to enable the backpropagation of gradients through the non-linear layers. Specifically, we implemented the backpropagation functionalities for rectified linear function (ReLU) and 2-dimensional average pooling layers. We also use truncation protocols to support the multiplication of fixed-point secret-shared numbers. For the concrete construction of these protocols, we refer the readers to Appendix~\ref{appendix:nonlinear}.

\parab{ReLU function} $\relu(x) = \max\{0, x\}$ is an activation function widely used in neural networks. 
In MPC, $\relu$ is effectively implemented by a composition of the derivative $\drelu(x) = \mathbf{1}\{x > 0\}$ and a multiplication, as $\relu(x) = \drelu(x) \cdot x$.

To support backpropagation for $\relu$, we need to evaluate 
\begin{equation}
    \label{eq:relu-bp}
    \grad{x} = \drelu(x) \cdot \grad{y}.
\end{equation}
We notice that the $\drelu(x)$ is already calculated in the forward propagation. Its result is secret shared in boolean form between the two parties as $\angleshare{d}_0 \oplus \angleshare{d}_1 = d = \drelu(x), \angleshare{d}_0, \angleshare{d}_1 \in \{0, 1\}$ where $\oplus$ is logical XOR. Therefore, we store this result in the forward propagation and reuse it in backpropagation to produce the secret shares of $\grad{x}$ in Equation (\ref{eq:relu-bp}).

\parab{2D Average Pooling Layer} At a high level, the 2D average pooling layer with kernel size $s$ outputs mean value of every adjacent $s\times s$ pixels. We use the division protocol provided by the framework of \cite{Huang2022cheetah} to support the forward and backward propagation of average pooling layers.

\parab{Truncation} Since we use the fixed-point representation in $\mathbb{Z}_t$, to avoid overflow, we need to reduce the precision from $2f$ to $f$ bits after every multiplication. In the forward propagation, we truncate the multiplication results of fully connected or convolutional layers after the activation function $\relu$. In the backward propagation, the gradients of intermediate outputs $\grad{\vector{X}}$ are truncated after being propagated through any linear layer.

\subsection{Preprocessing Optimization}\label{subsec:preprocessing-optimization}

In both forward and backward propagation, we notice that the substantial part of computation lies in the plaintext-ciphertext evaluation of linear operations $\circ$ and $\odot$, \ie $\vector{W} \circ \enc{\vector{X}}$ in Algorithm~\ref{alg:homomorphic-linear-evaluation}, and $\enc{\grad{\vector{W}}^{\mathsf{cross}}} = \share{\grad{\vector{Y}}}_0 \odot \enc{\share{\vector{X}}_1} + \enc{\share{\grad{\vector{Y}}}_1} \odot \share{\vector{X}}_0$ in Algorithm~\ref{alg:weight-gradient-calculation}. These evaluations are performed in every training iteration. We propose an optimization method to reduce the total number of such evaluations required in training. Specifically, originally we need $O(T)$ (number of training iterations) online evaluations, while the optimized approach only performs $O(m^2)=O(1)$ offline evaluations ($m$ is a constant agreed by the two parties) and is completely free of HE computation in the online phase.

We start by constructing a general protocol $\productProtocol(\circ, \vector{u}, \vector{v})$ for calculating the shares of $\vector{u} \circ \vector{v}$ for any linear operator $\circ$, where $\vector{u}$ and $\vector{v}$ are private data owned by \mowner and \downer, respectively.

\parab{For Fixed $\vector{u}$ and Variable $\vector{v}$}
We first consider a fixed $\vector{u}$. Inspired by a series of art~\cite{Beaver1992Efficient, Mohassel2017secureml, Mishra2020delphi}, we observe that: if for some random $\vector{v}'$, the product $\share{\vector{u} \circ \vector{v}'}$ could be evaluated and shared beforehand, then given the real $\vector{v}$, two parties can compute $\share{\vector{u} \circ \vector{v}}$  without HE at all. This protocol can be summarized as follows:

\begin{itemize}[leftmargin=*]
    \item \textbf{Preprocessing phase} (prepares shares of $\share{\vector{u}\circ\vector{v}'}$):
    \begin{itemize}
        \item[(1)] \downer chooses random mask $\vector{v}' \sim \vector{v}$ and sends $\enc{\vector{v}'}$;
        \item[(2)] \mowner chooses random mask $\vector{s} \sim \vector{u} \circ \vector{v}$, evaluates $\enc{\vector{u} \circ \vector{v}' - \vector{s}} = \vector{u} \circ \enc{\vector{v}'} - \vector{s}$ and sends it back. Thus, the two parties get shares of $\vector{u} \circ \vector{v}'$: $\share{\vector{u} \circ \vector{v}'}_0 = \vector{s}, \share{\vector{u} \circ \vector{v}'}_1 = \vector{u} \circ \vector{v}' - \vector{s}$.
    \end{itemize}
    \item \textbf{Online phase} (produces shares of $\share{\vector{u}\circ\vector{v}}$):
    \begin{itemize}
        \item[(3)] \downer sends masked $\vector{v} - \vector{v}'$ to \mowner, outputs $\share{\vector{u} \circ \vector{v}}_1 = \share{\vector{u} \circ \vector{v}'}_1;$
        \item[(4)] \mowner calculates and outputs
        $\share{\vector{u} \circ \vector{v}}_0 = \vector{u} \circ (\vector{v} - \vector{v}') + \share{\vector{u} \circ \vector{v}'}_0.$
    \end{itemize}
\end{itemize}

Although this protocol could offload HE operations to the preprocessing phase for \emph{a single evaluation} of $\vector{u} \circ \vector{v}$, we cannot extend it to \emph{multiple evaluations} of $\vector{u} \circ \vector{v}_i$ with different $\vector{v}_i$. Specifically, if we used the same mask $\vector{v}'$ for two different $\vector{v}_1$ and $\vector{v}_2$, then in the online phase the \downer would receive $\vector{v}_1 - \vector{v}'$, and $\vector{v}_2 - \vector{v}'$. Thus the difference $\vector{v}_1 - \vector{v}_2$ would be leaked to \mowner. 

To address this issue, we propose using multiple masks $\vector{v}'_i$. In the preprocessing phase, the two parties generate $m$ shared product $\vector{u} \circ \vector{v}'_i$ for different $\vector{v}'_i, i \in [m]$. In the online phase, for input $\vector{v}$, the \downer chooses $m$ non-zero scalars $k_i$ and sends a masked version of $\vector{v}$,
\begin{equation}
\tilde{\vector{v}} = \vector{v} - \sum_{i\in[m]} k_i \cdot \vector{v}'_i
\end{equation}
to \mowner. The two parties output
\begin{equation}
    \begin{aligned}
        \share{\vector{u} \circ \vector{v}}_0 & = \vector{u} \circ \tilde{\vector{v}} + \sum_{i\in[m]} k_i \cdot \share{\vector{u} \circ \vector{v}'_i}_0 \\
        \share{\vector{u} \circ \vector{v}}_1 & = \sum_{i\in[m]} k_i \cdot \share{\vector{u} \circ \vector{v}'_i}_1 \\
    \end{aligned}
\end{equation}

\parab{For Variable $\vector{u}$} 
Now we consider the case where $\vector{u}$ is a variable (\ie not determined at the preprocessing phase). Similar to \downer, \mowner also masks its $\vector{u}$ with $m$ masks $\vector{u}'_i$. The online phase could be viewed as a two-step process: \first the two parties generate shares $\angleshare{\vector{u} \circ \vector{v}'_j}$ for all $j \in [m]$; \second they apply these shares to further produce shares $\share{\vector{u} \circ \vector{v}}$. We present the full protocol in Algorithm~\ref{alg:preprocessing-technique}. In practice, Step~\ref{alg:preprocessing-technique:step:online-mowner-send} and Step~\ref{alg:preprocessing-technique:step:online-downer-send} in Algorithm~\ref{alg:preprocessing-technique} can be merged in parallel, as they are independent. 

\begin{algorithm}[t]
    \caption{$\productProtocol(\circ, \vector{u}, \vector{v})$: Preprocessing optimization for calculating the shares of $\vector{u} \circ \vector{v}$}
    \label{alg:preprocessing-technique}
    \KwIn{A predefined linear operation $\circ$; in the online phase, \mowner inputs $\vector{u}$ and \downer inputs $\vector{v}$.}
    \KwOut{The two parties receive shares of $\share{\vector{u} \circ \vector{v}}$.}
    \setlength{\belowdisplayskip}{0.3ex} \setlength{\belowdisplayshortskip}{0.3ex}
    \setlength{\abovedisplayskip}{0.3ex} \setlength{\abovedisplayshortskip}{0.3ex}
    \textbf{Preprocessing $\productProtocolPrep(\circ)$:} \\
    \setlength{\leftskip}{1em}
    \mowner selects $m$ random masks $\vector{u}'_i \sim \vector{u}, i \in [m]$; \\
    \downer selects $m$ random masks $\vector{v}'_j \sim \vector{v}, j \in [m]$, and sends their encryption $\enc{\vector{v}'_j}$ to \mowner;  \\
    \mowner selects $m^2$ masks $\vector{s}_{ij} \sim (\vector{u} \circ \vector{v}), i,j\in [m]$; \\
    \mowner evaluates $\enc{\angleshare{\vector{u}'_i \circ \vector{v}'_j}_1} = \vector{u}'_i \circ \enc{\vector{v}'_j} - \vector{s}_{ij}$ for $i,j \in [m]$, and sends them back for decryption; \\
    \mowner and \downer keeps shares of $\angleshare{\vector{u}'_i \circ \vector{v}'_j}$ for all $i, j \in [m]$
    $$
        \begin{aligned}
            \angleshare{\vector{u}'_i \circ \vector{v}'_j}_0 & = \vector{s}_{ij} \\
            \angleshare{\vector{u}'_i \circ \vector{v}'_j}_1 & = \vector{u}'_i \circ \vector{v}'_j - \vector{s}_{ij}
        \end{aligned}
    $$ \\
    \setlength{\leftskip}{0em}
    \textbf{Online $\productProtocolOnline(\circ, \vector{u}, \vector{v})$:} \\
    \setlength{\leftskip}{1em}
    \mowner randomly picks scalars $k_i, i \in [m]$; \mowner sends to \downer all $k_i$ and \label{alg:preprocessing-technique:step:online-mowner-send}
    $$\tilde{\vector{u}} = \vector{u} - \sum_{i\in[m]} k_i \cdot \vector{u}'_i$$ \\
    \mowner and \downer produces shares of $\angleshare{\vector{u} \circ \vector{v}'_j}$ for all $j \in [m]$ as \label{alg:preprocessing-technique:step:share-fixed-u}
    $$
        \begin{aligned}
            \angleshare{\vector{u} \circ \vector{v}'_j}_0 & = \sum_{i \in [m]} k_i \cdot \angleshare{\vector{u}'_i \circ \vector{v}'_j}_0 \\
            \angleshare{\vector{u} \circ \vector{v}'_j}_1 & = \tilde{\vector{u}} \circ \vector{v}'_j + \sum_{i \in [m]} k_i \cdot \angleshare{\vector{u}'_i \circ \vector{v}'_j}_1
        \end{aligned}
    $$ \\
    \downer randomly picks scalars $\ell_j, j \in [m]$; \downer sends to \mowner all $\ell_j$ and \label{alg:preprocessing-technique:step:online-downer-send}
    $$\tilde{\vector{v}} = \vector{v} - \sum_{j\in[m]} \ell_j \cdot \vector{v}'_j$$ \\
    \mowner and \downer produces shares of $\angleshare{\vector{u} \circ \vector{v}}$ as \label{alg:preprocessing-technique:step:finish} 
    $$
        \begin{aligned}
            \angleshare{\vector{u} \circ \vector{v}}_0 & = \vector{u} \circ \tilde{\vector{v}} + \sum_{j \in [m]} \ell_j \cdot \angleshare{\vector{u} \circ \vector{v}'_j}_0 \\
            \angleshare{\vector{u} \circ \vector{v}}_1 & = \sum_{j \in [m]} \ell_j \cdot \angleshare{\vector{u} \circ \vector{v}'_j}_1
        \end{aligned}
    $$
\end{algorithm}

For simplicity, we denote Algorithm~\ref{alg:preprocessing-technique} as a general protocol $$\productProtocol(\circ, \vector{u}, \vector{v}) = (\productProtocolPrep(\circ), \productProtocolOnline(\circ, \vector{u}, \vector{v})),$$ consisting of a preprocessing protocol $\productProtocolPrep(\circ)$ independent of $\vector{u}, \vector{v}$ and a online protocol $\productProtocolOnline(\circ, \vector{u}, \vector{v})$ to specifically evaluate $\vector{u} \circ \vector{v}$.

\parab{Applying $\productProtocol(\circ, \vector{u}, \vector{v})$ in Linear Layer Training}
We now apply $\productProtocol(\circ, \vector{u}, \vector{v})$ to accelerate the online training of linear layers in Algorithm~\ref{alg:homomorphic-linear-evaluation} and Algorithm~\ref{alg:weight-gradient-calculation}. 
In particular, for every invocation with the same $\circ$ of $\productProtocol(\circ, \vector{u}, \vector{v})$, the preprocessing phase $\productProtocolPrep(\circ)$ is executed once and for all, before the training starts. When $\vector{u}, \vector{v}$ come on the fly during training, the two parties only execute $\productProtocolOnline(\circ, \vector{u}, \vector{v})$. 
We present the training protocol optimized by $\productProtocol(\circ, \vector{u}, \vector{v})$ in Algorithm~\ref{alg:training-iteration-with-preprocessing}. For consistency of symbols, we semantically denote $\vector{v} \odot_\reverse \vector{u} = \vector{u} \odot \vector{v}$. We provide the security analysis of this preprocessing optimization technique in \S~\ref{sec:privacy-analysis}.

\begin{algorithm}[t]
    \caption{Optimized Training Protocol of Linear Layers} 
    \label{alg:training-iteration-with-preprocessing}
    \setlength{\belowdisplayskip}{0.3ex} \setlength{\belowdisplayshortskip}{0.3ex}
    \setlength{\abovedisplayskip}{0.3ex} \setlength{\abovedisplayshortskip}{0.3ex}
    \textbf{Preprocessing:} \\
    \setlength{\leftskip}{1em}
    The two parties invoke $\productProtocolPrep(\circ)$, $\productProtocolPrep(\odot_x)$, $\productProtocolPrep(\odot)$, $\productProtocolPrep(\odot_\reverse)$. \\
    \setlength{\leftskip}{0em}
    \textbf{Online:} \\

    \emph{Foward propagation}: \mowner, \downer provide shares of $\share{\vector{X}}$. \\
    
    \setlength{\leftskip}{1em}
    The two parties invoke $\productProtocolOnline(\circ, \vector{W}, \share{\vector{X}}_1)$ to produce the shares of $\angleshare{\vector{W}\circ \angleshare{\vector{X}}_1}$. \mowner and \downer respectively output
    $$
        \begin{aligned}
            \angleshare{\vector{Y}}_0 & = \angleshare{\vector{W}\circ \angleshare{\vector{X}}_1}_0 + \vector{W} \circ \angleshare{\vector{X}}_0 + \vector{b} \\
            \angleshare{\vector{Y}}_1 & = \angleshare{\vector{W}\circ \angleshare{\vector{X}}_1}_1
        \end{aligned}
    $$ \\

    \setlength{\leftskip}{0em}
    \emph{Backpropagation}: \mowner, \downer provide shares of $\share{\grad{\vector{Y}}}$. \\
    
    \setlength{\leftskip}{1em}
    (For $\grad{\vector{X}}$) The two parties invoke $\productProtocolOnline(\odot_x, \vector{W}, \angleshare{\grad{\vector{Y}}}_1)$ to produce shares of $\vector{W}\odot_x \angleshare{\grad{\vector{Y}}}_1$. \mowner and \downer respectively output
    $$
        \begin{aligned}
            \angleshare{\grad{\vector{X}}}_0 & = \angleshare{\vector{W}\odot_x \angleshare{\grad{\vector{Y}}}_1}_0 + \vector{W} \odot_x \angleshare{\grad{\vector{Y}}}_0 \\
            \angleshare{\grad{\vector{X}}}_1 & = \angleshare{\vector{W}\odot_x \angleshare{\grad{\vector{Y}}}_1}_1
        \end{aligned}
    $$ \\

    (For $\grad{\vector{b}}$) The two parties locally compute $\angleshare{\grad{\vector{b}}}_0$ and $\angleshare{\grad{\vector{b}}}_1$. \downer add perturbation to its share, and $\grad{\vector{b}}$ is revealed to \mowner. \\

    (For $\grad{\vector{W}}$) The two parties invoke $\productProtocolOnline(\odot, \angleshare{\grad{\vector{Y}}}_0, \angleshare{\vector{X}}_1)$, $\productProtocolOnline(\odot_\reverse, \angleshare{\vector{X}}_0, \angleshare{\grad{\vector{Y}}}_1)$ to produce shares of the cross terms $\angleshare{\grad{\vector{W}}^{\mathsf{01}}}$ and $\angleshare{\grad{\vector{W}}^{\mathsf{10}}}$:
    $$
        \begin{aligned}
            \angleshare{\grad{\vector{W}}^{\mathsf{01}}} & = \angleshare{\angleshare{\grad{\vector{Y}}}_0 \odot \angleshare{\vector{X}}_1} \\ 
            \angleshare{\grad{\vector{W}}^{\mathsf{10}}} & = \angleshare{\angleshare{\grad{\vector{Y}}}_1 \odot \angleshare{\vector{X}}_0} = \angleshare{\angleshare{\vector{X}_0} \odot_\reverse \angleshare{\grad{\vector{Y}}}_1} \\ 
        \end{aligned}
    $$ \\
    \downer calculates
    $$\widehat{\grad{\vector{W}}} = \angleshare{\grad{\vector{W}}^{\mathsf{01}}}_1 + \angleshare{\grad{\vector{W}}^{\mathsf{10}}}_1 + \angleshare{\grad{\vector{Y}}}_1 \odot \angleshare{\vector{X}}_1$$ \\
    \downer adds a perturbation $\vector{e}$ to $\widehat{\grad{\vector{W}}}$ and sends it to \mowner; \\
    \mowner finishes with
    $$\grad{\vector{W}} = \widehat{\grad{\vector{W}}} + \angleshare{\grad{\vector{W}}^{\mathsf{01}}}_0 + \angleshare{\grad{\vector{W}}^{\mathsf{10}}}_0 + \angleshare{\grad{\vector{Y}}}_0 \odot \angleshare{\vector{X}}_0$$
\end{algorithm}

\subsection{Extending to Multiple \downers}

In practice, the dataset provided by each \downer could be highly heterogeneous or even biased. Thus, it is desirable to simultaneously train a model using the combined data contributed by different \downers. 

As stated in \S~\ref{subsec:intro:related-work}, extensibility is challenging in prior art. For instance, in previous HE-based training protocols (\cite{Nandakumar2019towards, Tian2022sphinx}), the \mowner can only train the model with one \downer, because the model is encrypted by the \downer. 
In MPC-based approaches (\cite{Mohassel2017secureml, Mohassel2018aby3, Chaudhari2020trident, Koti2021Tetrad}), if \mowner and \downers upload their model and data to several fixed computing servers, it would introduce the undesirable ``non-colluding'' assumption. On the other hand, if the \mowner and all \downers themselves participate as computing nodes, the MPC scheme would have to defend privacy against up to $n-1$ colluding parties, which would result in low efficiency.

Our framework avoids the high overhead of directly supporting an $n$-party computation by decomposing the procedure into the 2-PC paradigm, as we only need interaction between the \mowner and one \downer at a time. The weights are kept by the \mowner in plaintext so it could simply conduct collaborative training with each \downer in turn and updates its model incrementally. This design makes \sys extend to more \downers without extra computation or communication overhead, unlike previous general $n$-party MPC methods. In \S~\ref{subsec:evaluation:multiple-does}, we evaluate the extensibility of \sys.

\section{Security Analysis}\label{sec:privacy-analysis}


\subsection{Security of the \sys Training Framework}
In this section, we first perform the security analysis for the \sys training framework without the preprocessing optimization described in \S~\ref{subsec:preprocessing-optimization}.  

\begin{definition}
    \label{def:private-training-protocol}
    A protocol $\Pi$ between a \downer possessing a train dataset $\mathcal{D} = \{(\vector{x}_i, t_i)\}$ and a \mowner possessing the model weights $\vector{M}$ is a \textbf{cryptographic training protocol} if it satisfies the following guarantees.
    \begin{itemize}[leftmargin=*, parsep=0.2ex, topsep=0.2ex, partopsep=0.2ex]
        \item \textbf{Correctness.} On every set of model weights $\vector{M}$ of the \mowner and every dataset $\mathcal{D}$ of the \downer, the output of the \mowner is a series of weight updates and finally a correctly trained model with updated weights.
        \item \textbf{Security.} 
        \begin{itemize}
            \item[-] (\textbf{Data privacy}) We require that a corrupted, semi-honest \mowner does not learn anything useful about the \downer's training data, \emph{except the weight updates and the final model}. Formally, we require the existence of an efficient simulator $\simulator_{MO}$ such that $\mathsf{View}_{MO}^{\Pi} \approx_c \simulator_{MO}(\vector{M}, \mathsf{out})$, where $\mathsf{View}_{MO}^{\Pi}$ denotes the view of the \mowner in the execution of $\Pi$, $\mathsf{out}$ denotes the output of the training protocol to \mowner, and $\approx_c$ denotes computational indistinguishability between two distributions.
            \item[-] (\textbf{Model privacy}) We require that a corrupted, semi-honest \downer does not learn anything useful about the \mowner's model weights, \emph{except the model's outputs (predictions) on \downer's dataset}. The model architecture is public to all parties. Formally, we require the existence of an efficient simulator $\simulator_{DO}$ such that $\mathsf{View}_{DO}^{\Pi} \approx_c \simulator_{DO}(\mathcal{D})$, where $\mathsf{View}_{DO}^{\Pi}$ denotes the view of the client in the execution of $\Pi$.
        \end{itemize} 
    \end{itemize}
\end{definition}

\begin{theorem}
    \label{theorem:basic-protocols}
    Assuming the existence of oblivious transfer, homomorphic encryption and secure protocols for non-linearity evaluations, the \sys framework without the preprocessing optimization is a cryptographic training protocol as defined in Definition~\ref{def:private-training-protocol}.
\end{theorem}

We rigorously prove Theorem~\ref{theorem:basic-protocols} using the real/ideal word paradigm by constructing  simulators for the \downer and the \mowner. Due to space constraint, the detailed simulator construction and hybrid proof are deferred to Appendix~\ref{appendix:security-proofs}.
Extending the theorem to the single-\mowner-multi-\downers scenario is trivial. 

\subsection{Distinguishability Caused by Prepossessing Optimization}

After introducing the preprocessing technique (see \S~\ref{subsec:preprocessing-optimization}) in \sys, the computational indistinguishability of these views no longer holds. This is because we use linear combinations of $m$ masks for the multiplication operands $\vector{u}$ and $\vector{v}$ of any linear operator $\circ$, instead of using uniformly random tensor masks. 
In this section, we analyze and quantify such distinguishability and prove that the privacy loss caused by the distinguishability is negligible (\ie it is computationally difficult for an adversary to derive private information based on the distinguishability).  

We first give a useful proof gadget. 
\begin{definition}
    A set $\vector{V} = \{\vector{v}_i\}_{i\in[n]}$ of $n$ elements is \textbf{$m$-linear combinatorially private} to a party \textsf{S}, if any property about the elements in $\vector{V}$, derived by \textsf{S}, has the form of a linear combination 
    \begin{equation}
        \label{eq:v-linear-combination}
        \sum_{i \in [m']} a_i \vector{v}_{n_i} = \hat{\vector{v}}
    \end{equation}
    where $n_i$ are $m'$ distinct indices in $[n]$ and $m'>m$. $a_i \neq 0$, $m$, and  $\hat{\vector{v}}$ are public parameters that are \emph{not controlled} by \textsf{S}.
\end{definition}

Now we present the following theorem.  

\begin{theorem}
    \label{theorem:processing-protocol-linear-combination}
    If $\productProtocolOnline(\circ, \vector{u}, \vector{v})$ (Algorithm~\ref{alg:preprocessing-technique}) is executed $n$ times with $n$ different $\vector{u}_i$ provided by \mowner and $n$ different $\vector{v}_i$ provided by \downer, and $m < n <\!\!< t$ ($t$ is the number of elements in the ring $\mathbb{Z}_t$), then $\vector{U}=\{\vector{u}_i\}_{i\in[n]}$ are $m$-linear combinatorially private to \downer, and $\vector{V} =\{\vector{v}_i\}_{i\in[n]}$ are $m$-linear combinatorially private to \mowner.
\end{theorem}

\begin{proof}
    Since $\vector{u}_i$ and $\vector{v}_i$ are symmetric for \downer and \mowner,  respectively, proving the theorem for $\vector{v}_i$ is sufficient. 
    We prove the theorem using the elimination of the masks: since each $\vector{v}_i$ is masked with $m$ different masks, eliminating every mask requires a new equation of $\tilde{\vector{v}}_i$ (see Equation~\ref{eq:masked-vi} below). Thus, the adversary could eventually obtain a linear combination of at least $m+1$ different $\vector{v}_i$'s in $\vector{V}$. The formal formulation is as follows.
    
    In the preprocessing phase $\productProtocolPrep(\circ)$, \downer spawns $m$ random $\vector{v}'_j$. In the online phase $\productProtocolOnline(\circ, \vector{u}_i, \vector{v}_i)$, \downer chooses $m$ scalars $\ell_{ij} \in \integers_t$ for each $\vector{v}_i$ and sends to \mowner the following 
    \begin{equation}
        \label{eq:masked-vi}
        \tilde{\vector{v}}_i = \vector{v}_i - \sum_{j=0}^{m-1} \ell_{ij} \vector{v}'_j, i\in [n].
    \end{equation}
    Let matrix $\vector{L} = (\ell_{ij}) \in \integers_t^{n \times m}$. Since $\ell_{ij}$ are chosen uniformly in $\integers_t$ and $n <\!\!< t$, $\vector{L}$ is full-rank ($\mathsf{rank}(\vector{L}) = m$) with high probability~\cite{Blomer1997TheRank, Cooper2000OnTheDistribution}.\footnote{Note \downer can intentionally set $\ell_{ij}$ to ensure $\vector{L}$ is full-rank. For example, \downer can use the Vandemonde matrix with $\ell_{ij} = j^i \pmod{t}$.} Therefore, in order to eliminate $\vector{v}'_j$ which are unknown to \mowner, \mowner needs at least $m+1$ equations in the form of Equation~(\ref{eq:masked-vi}) to obtain a linear combination of elements in $\vector{V}$. Finally, \mowner obtains a relation of $m'>m$ tensors in $\{\vector{v}_i\}$ of its choice:
    \begin{equation*}
        \sum_{k=0}^{m'-1} a_k\vector{v}_{n_k} = \hat{\vector{v}},
    \end{equation*} 
    where the linear combination coefficients $a_k$ are determined by $\vector{L}$, and $\hat{\vector{v}}$ could be calculated with the knowledge of $\tilde{\vector{v}}_{n_i}$. By definition, $\{\vector{v}_i\}$ are $m$-linear combinatorially private to \mowner.
\end{proof}

\begin{corollary}
    \label{corollary:search-space}
    By setting the appropriate number of masks $m$ and fixed-point bit precision $f$ in \sys, it is computationally difficult for an adversary to derive the elements in a $m$-linear combinatorially private set $\vector{V}$, because the adversary needs to exhaustively explore a search space of size $O(2^{fm})$.
\end{corollary}

\emph{Proof sketch.} Suppose an adversary tries to obtain one specific element $\vector{v}_i$ in $\vector{V}$ from the linear combination $\sum_{k\in[m']} a_k\vector{v}_k = \hat{\vector{v}}$ in Eq. (\ref{eq:v-linear-combination}). For simplicity, we let $n_k = k$ and $\vector{v}_i = v_i$ has only one dimension.  The adversary cannot use iterative methods such as gradient descent to approximate the solution, since the ring $\integers_t$ is discrete. Instead, supposing that the decimal value $\tilde{v}_i = \mathsf{Encode}^{-1}(v_i)$ \footnote{Recall $\mathsf{Encode}(\tilde{v})$ is the fixed-point encoding function from $\mathbb{R}$ to $\mathbb{Z}_t$ (\S~\ref{subsec:prelim:additive-secret-sharing}).} has a range $[-b/2, b/2]$, $b=O(1)$, the adversary needs to search $m'-1\geq m$ variables and check if the last variable is decoded into the range. In particular, the adversary solves the following problem: 
$$\begin{aligned}
    \{v_i \in \integers_t \}_{i\in[m'-1]}
    ~ & \text{s.t.} ~ v_{m'-1} = a_{m'-1}^{-1}(\hat{v} - \sum_{i=0}^{m'-2}a_kv_k) \\
    ~ & \text{and} ~ \mathsf{Encode}^{-1}(v_i) \in [-b/2,b/2], \forall i\in[m']
\end{aligned}$$
Each variable takes $b\cdot 2^f$ possible values in $\integers_t$, resulting in a total search space of  $O(2^{fm})$. 

In \S~\ref{subsec:evaluation:gradient-attacks}, we experimentally show that by setting $m=8$ and $f=25$, \emph{deriving the elements in a set $\vector{V}$ that is $m$-linear combinatorially private is as difficult as deciphering elements encrypted by 7680-bit RSA keys}. 

\subsection{Privacy Analysis of the Weight Updates}
In Step~\ref{alg:step:weight-dp} of Algorithm~\ref{alg:weight-gradient-calculation}, we propose a secret-shared version of DPSGD~\cite{Abadi2016dpsgd} that enables the \downer to add a perturbation $\vector{e} \leftarrow \mathcal{N}\left(0, \frac{\sigma^2 C^2}{B} \mathbb{I}\right)$ to the weight updates. Under this setting,  
there exist constants $c_1, c_2$, such that given batch size $B$, training dataset size $N$ and number of training steps $T$, for any $\epsilon<c_1B^2T/N^2$, this mechanism is $(\epsilon, \delta)$-differentially private for any $\delta > 0$, if we choose
\begin{equation}
    \sigma \geq c_2\frac{B\sqrt{T\log(1/\delta)}}{N\epsilon}
\end{equation}

\section{Hardware Acceleration}\label{sec:hardware_acc}

In this section, we present our GPU acceleration of BFV cryptosystem by exploiting the parallelizable feature of lattice-based homomorphic cryptosystem. 
Our design is independent of the protocols in \sys and universally applicable. Typically, GPU provides advantage over CPU in the case where multiple  operations are parallelizable. Therefore, we try to recognize which operations in the BFV construction  are parallelizable, and to what extent they can be parallelized.

\parab{NTT} The BFV scheme operates on the polynomial ring. Although polynomial addition is trivial with each coefficient added individually, the multiplication of polynomials (modulo $X^N+1$) cannot be done in linear time trivially. Direct multiplication requires $O(N^2)$ time complexity. Number Theory Transform (NTT) addresses this issue using a bijection between the polynomial ring $\ring_{q,N}$ and the vector space $\mathbb{Z}_q^N$. With this bijection, the multiplication on the ring corresponds to vector element-wise multiplication on $\mathbb{Z}_q^N$. Therefore, when a polynomial is represented in the NTT form, both addition and multiplication of polynomials receive a degree of parallelism of $N$.

\parab{RNS Decomposition} The ciphertext coefficient modulus $q$ of BFV scheme could be very large (\eg $>2^{160}$). Thus, implementing the vector operations directly in $\mathbb{Z}_q$ requires the arithmetic of big integers, which is not natively supported by modern CPUs and GPUs. \cite{Bajard2017RNSBFV} proposes to improve efficiency by applying the Chinese Remainder Theorem: If $q = q_1q_2 \cdots q_L$ with each $q_i < 2^{64}$, by an isomorphism, the field $\mathbb{Z}_{q} \cong \mathbb{Z}_{q_1} \times \mathbb{Z}_{q_2} \times \cdots \times \mathbb{Z}_{q_L}$. Therefore, an element $a \in \mathbb{Z}_q$ could be represented by a vector $\vector{a} = (a\mod q_1, a\mod q_2, \cdots, a \mod q_L)$. The addition and multiplication are conveyed to each smaller field $\mathbb{Z}_{q_i}$, which can be effectively implemented with modern processors supporting $64$-bit integer arithmetic. This technique is called residual number system (RNS) decomposition.

\parab{Parallelization} With NTT transform and RNS decomposition, the arithmetic of BFV scheme is conveyed into $NL$-dimensional vector space $\prod_{i=1}^L \mathbb{Z}_{q_i}^N$. Addition and multiplication in message space (polynomial ring $\ring_{t,N}$) corresponds to vector addition and element-wise multiplication. Typically, the polynomial degree $N  \geq 4096$, and the number of decomposed $q_i$ is $L \geq 2$. Therefore, the homomorphic operations have a degree of parallelism of at least $8192$. We exploit this parallelism with an implementation on GPU. To fully support the BFV scheme, we also provide efficient GPU implementations of NTT and RNS composition/decomposition, besides vector addition and multiplication.

\parab{Memory management} We notice that allocating and freeing GPU memory could be time-consuming if they are executed with every construction and disposal of a ciphertext/plaintext. Therefore, we use a memory pool to keep track of each piece of allocated memory, only freeing them when the program exits or when the total memory is exhausted. The memory of discarded objects is returned to the pool. When a new ciphertext or plaintext is required, the memory pool tries to find a piece of memory already allocated with a suitable size. If not found, allocation is performed.

\section{Evaluation}


Our evaluations are designed to demonstrate the following. 

\begin{itemize}[leftmargin=*, nosep]
    \item \textbf{End-to-End training performance}. We show that \sys is able to train models (both from scratch and via transfer learning) with test accuracies \emph{nearly identical to plaintext training}, and could be extended to multiple \downers.
    \item \textbf{Performance breakdown}. We evaluate the efficiency of the key training protocols in \sys, as well as the performance gains of accelerating individual HE operations using our hardware implementation.  
    \item \textbf{Efficiency comparison with prior art}. We show that \sys has non-trivial performance (\eg training time and communication bytes) advantages over closely related art (although none of them offer independent model deployability and native extensibility to include more \downers)
    \item \textbf{\sys against attacks}. Finally, we experimentally demonstrate that \sys is secure against both existing attacks and adaptive attacks specifically designed to exploit the protocols of \sys. 
\end{itemize}

\subsection{Implementation}

We implement the training protocols of \sys primarily in Python, where we formalize the common layers of neural networks as separate modules. We use C++ and the CUDA programming library to implement the GPU version of the BFV scheme, referring to the details specified in state-of-the-art CPU implementation, Microsoft SEAL library~\cite{SEAL}. For non-linear layers, we adapt the C++-based framework of OpenCheetah~\cite{Huang2022cheetah} and CrypTFlow2~\cite{Rathee2020cryptflow2}, which supports efficient evaluation of the ReLU function, division and truncation, etc. Pybind11~\cite{pybind11} is used to encapsulate the C++ interfaces for Python. The total developmental effort is $\sim$15,000 lines of code. Our open-sourced code artifact is described in \S~\ref{appendix:artifact}. 


\subsection{Evaluation Setup}\label{subsec:evaluation:setup}

We evaluate our framework on a physical machine with Intel Xeon Gold 6230R CPU and NVIDIA RTX A6000 GPU (CUDA version 11.7), under the operating system Ubuntu 20.04.5 LTS. We evaluate two typical network settings: a LAN setting with 384MB/s bandwidth and a WAN setting with 44MB/s bandwidth, same as \cite{Huang2022cheetah}. 
BFV HE scheme is instantiated with $N=8192$, $t = 2^{59}$, and $q \approx 2^{180}$ decomposed into three RNS components. This HE instantiation provides $\lambda = 128$ bits of security. The fixed-point precision is $f=25$.

We use the following three datasets in our evaluations:
\begin{itemize}[leftmargin=*, nosep]
    \item MNIST dataset~\cite{Lecun1998Gradient} includes 70,000 grayscale images (60,000 for training, 10,000 for testing) of handwritten digits from 0 to 9, each with $28\times 28$ pixels.
    \item CIFAR10 dataset~\cite{Krizhevsky2009LearningML} includes 60,000 RGB images (50,000 for training, 10,000 for testing) of different objects from ten categories (airplane, automobile, bird, etc.), each with $32\times 32$ pixels.
    \item AGNews dataset \cite{Zhang2015Agnews} includes 120,000 training and 7,600 testing lines of news text for 4 different topic categories.
\end{itemize}

\subsection{End-to-End Training Performance}

We evaluate the performance of \sys in a series of different training scenarios.

\subsubsection{Training from Scratch} 
The first case studied is that a \mowner trains a model from scratch with one \downer's data. We experiment with 4 neural networks, denoted as MLP (multi-layer perception) for MNIST~\cite{Mohassel2018aby3, Patra2022blaze}, CNN for MNIST~\cite{Riazi2018Chameleon}, and TextCNN for AGNews~\cite{Zhang2015Agnews}, CNN for CIFAR10~\cite{Tian2022sphinx}. 
The detailed architectures of these NNs are listed in Appendix~\ref{appendix:nn-architecture}. We train each image NN for 10 epochs and the text NN for 5 epochs using the training dataset, and test the classification accuracy on the test dataset every 1/5 epoch. Batch size is 32 for MNIST and AGNews and 64 for CIFAR10. We use SGD optimizer with a momentum of 0.8, learning rate $\eta=10^{-2}$. 
For fair comparison with plaintext training, the \downer does not add DP noises to the weight gradients in these experiments (see \S~\ref{subsec:evaluation:dp-impact-on-accuracy} for results with DP noises). The results are shown in Table~\ref{table:training-accuracies} and the test accuracy curves are plotted in Figure~\ref{fig:acc-curve}. We observe that, the test accuracies of models trained in \sys experience very minor declines (on average $<0.7\%$) compared with plaintext training, which might be attributed to the limited precision in fixed-point arithmetic. Note that the model accuracies are relatively lower than SOTA models for the CIFAR10 dataset because \first the models used in this part are relatively small and \second we only train them for 10 epochs. 

\begin{table}
    \centering
    \begin{tabular}{ccc|rr}
        \hline
        Scenario & Task  & Model & \sys & Plaintext \\ \hline
        \multirow[c]{4}{*}{\tabincell{c}{Train from \\ scratch}}
        & MNIST           &  MLP & 97.74\% & 97.82\% \\
        & MNIST           &  CNN & 98.23\% & 98.59\% \\ 
        & AGNews          & TextCNN & 87.72\% & 87.97\% \\
        & CIFAR10         &  CNN & 71.69\% & 72.27\% \\ 
        \hline 
        \multirow[c]{2}{*}{\tabincell{c}{Transfer \\ learning}}
        & CIFAR10         & AlexNet & 86.72\% & 86.90\% \\ 
        & CIFAR10         & ResNet50 & 90.02\% & 89.87\% \\ 
        \hline
    \end{tabular}
    \caption{Highest test accuracy in 10 training epochs for different ML tasks, reported for training with \sys (private training) and plaintext training, respectively. \sys achieves nearly identical training accuracy as plaintext training.}
    \label{table:training-accuracies}
\end{table}

\begin{table*}
    \centering
    \small
    \begin{tabular}{ccc||rrr||rr|rrr}
        \hline

        \multirow[c]{3}{*}{Scenario} &
        \multirow[c]{3}{*}{Task} &
        \multirow[c]{3}{*}{Model} &
        \multicolumn{3}{c||}{\sys} &
        \multicolumn{5}{c}{\sysOpt} \\
        \cline{4-11}

        & & &
        \multicolumn{3}{c||}{Online} &
        \multicolumn{2}{c|}{Preprocessing} &
        \multicolumn{3}{c}{Online} \\

        & & &
        $\mathsf{TP}_{\mathsf{LAN}}$ & $\mathsf{TP}_{\mathsf{WAN}}$ & $\mathsf{C}$ &
        $\mathsf{T}_\mathsf{prep}$ & $\mathsf{C}_\mathsf{prep}$ & 
        $\mathsf{TP}_{\mathsf{LAN}}$ & $\mathsf{TP}_{\mathsf{WAN}}$ & $\mathsf{C}$
        \\ \hline

        \multirow[c]{4}{*}{\tabincell{c}{Train from \\ scratch}}
        & MNIST   &     MLP & $9.73\times10^4$ & $5.12\times10^4$ &  1.66 & 0.02 &  3.35 & $26.52\times10^4$ & $19.87\times10^4$ &  0.23 \\
        & MNIST   &     CNN & $7.70\times10^4$ & $4.43\times10^4$ &  1.71 & 0.02 &  4.13 & $13.72\times10^4$ & $10.75\times10^4$ &  0.36 \\
        & AGNews  & TextCNN & $0.37\times10^4$ & $0.53\times10^4$ & 14.62 & 0.27 & 19.28 & $ 0.76\times10^4$ & $ 1.07\times10^4$ &  6.74 \\
        & CIFAR10 &     CNN & $0.18\times10^4$ & $0.12\times10^4$ & 44.89 & 0.70 & 83.12 & $ 0.22\times10^4$ & $ 0.15\times10^4$ & 34.90 \\
        \hline
        \multirow[c]{2}{*}{\tabincell{c}{Transfer \\ learning}}
        & CIFAR10 & AlexNet & $0.52\times10^4$ & $0.39\times10^4$ & 11.33 & 0.91 & 46.00 & $ 1.55\times10^4$ & $ 1.24\times10^4$ &  2.90 \\
        & CIFAR10 & ResNet50 & $1.83\times10^4$ & $1.17\times10^4$ &  5.48 & 0.30 & 15.96 & $ 8.05\times10^4$ & $ 5.89\times10^4$ &  0.82 \\
        
        \hline
    \end{tabular}
    \caption{Training costs for different ML tasks. For the online phase, $\mathsf{TP}$ stands for the throughput (samples/hour) of the training system, and subscript $\mathsf{LAN}, \mathsf{WAN}$ indicate the network settings; $\mathsf{C}$ stands for the online communication (MB) per sample. For \sysOpt, we also report the time ($\mathsf{T}_\mathsf{prep}$, hours) and communication ($\mathsf{C}_\mathsf{prep}$, GB) of preprocessing. Note that the preprocessing overhead is one-time overhead.}
    \label{table:training-costs}
\end{table*}

\begin{figure}
    \centering
    \subfigure[MNIST, MLP]{
        \includegraphics[width=.45\linewidth]{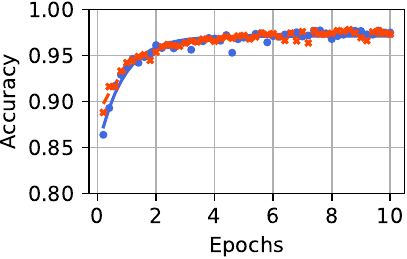}
    }
    \subfigure[MNIST, CNN]{
        \includegraphics[width=.45\linewidth]{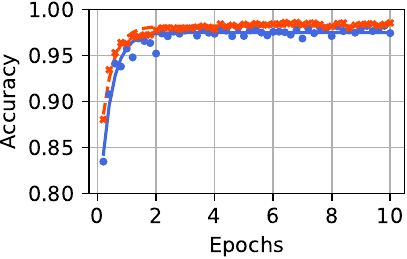}
    } \\
    \subfigure[AGNews, TextCNN]{
        \includegraphics[width=.45\linewidth]{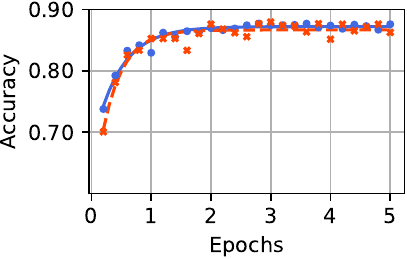}
    }
    \subfigure[CIFAR10, CNN]{
        \includegraphics[width=.45\linewidth]{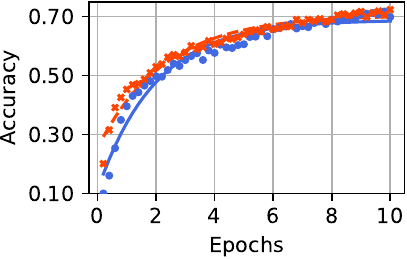}
    } \\
    \subfigure[CIFAR10, AlexNet]{
        \includegraphics[width=.45\linewidth]{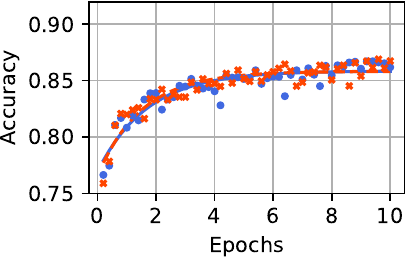}
    }
    \subfigure[CIFAR10, ResNet50]{
        \includegraphics[width=.45\linewidth]{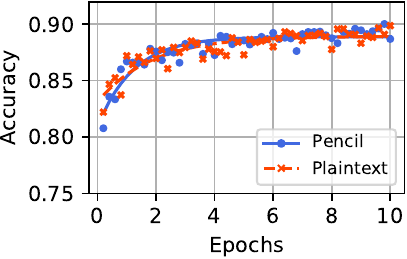}
    }
    \caption{Test accuracies for trained models. (a) $\sim$ (d) are for models trained from scratch; (e) and (f) are for models trained via transfer learning.
    }
    \label{fig:acc-curve}
\end{figure}

We also report the total time and communication needed to train these NNs for one epoch in Table~\ref{table:training-costs}. We measure the overheads for \sys both with and without preprocessing optimization, denoted as \sys and \sysOpt respectively. For the preprocessing phase, we set the number of masks to $m=8$. 
In the LAN network setting, \sys without preprocessing can train CNNs for MNIST and CIFAR10 from scratch (10 epochs) within 7.8 hours and 280 hours, respectively. When the preprocessing technique is enabled, the training times are further reduced to 4.4 hours and 229 hours, respectively.
In the WAN setting, we see a throughput reduction of $30\% \sim 50\%$, since the communication overhead is comparable to computation overhead. We observe that in training CNN for CIFAR10, the improvement of the preprocessing technique is not as significant as the previous 3 NNs. This is because a significant portion of training overhead of this NN is introduced by the non-linear evaluations (\eg ReLU), while the preprocessing technique focuses on linear layers. 

\subsubsection{Transfer Learning Models}
We now evaluate applying \sys in transfer learning: 
\mowner uses a publicly available pre-trained model as the feature extractor, and subsequently trains a classifier on top of it.
Since models pre-trained on large general datasets achieve robust feature extraction, transfer learning can significantly reduce the convergence time and improve test accuracy on specialized datasets, as shown by~\cite{Oquab2014Learning, Huh2016WhatMakes}.
Before training starts, both \downer and \mowner obtain the public feature extractor. During training, the \downer first passes its data through the feature extractor to get intermediate representations (IRs), and the classifier is trained privately on these IRs. 

In our experiments, we use two publicly available pre-trained models, AlexNet~\cite{Krizhevsky2017AlexNet} and ResNet50 \cite{He2016ResNet} as the feature extractor.
\footnote{The original model is split into two consecutive parts: the feature extractor and the original classifier on a general dataset. We replace the original classifier with several fully connected layers to be trained. The weights pre-trained on ImageNet \cite{Deng2009imagenet} are provided by the \texttt{torchvision} library.}
The complete models are denoted with the pre-trained model's name (see Appendix~\ref{appendix:nn-architecture} for the detailed architecture). The results are shown in Table~\ref{table:training-accuracies} and Figure~\ref{fig:acc-curve}. The best accuracy for CIFAR10 is improved to 90.02\%, compared with the train-from-scratch CNN (71.69\%). Meanwhile, the training time and communication overhead are also greatly reduced, as shown in Table~\ref{table:training-costs}. For instance, transfer learning of ResNet50 sees a $4.4\times$ boost in throughput compared to training-from-scratch CNN, and can be trained within 6.2 hours in the LAN setting.

\emph{Key takeaways}: for relatively large models like ResNet50 (with 25 million parameters), the overhead of cryptographic training from scratch is still significant. Augmented with transfer learning, cryptographic training can be more practical. 

\subsubsection{Training with Heterogeneous \downers}\label{subsec:evaluation:multiple-does}

In this segment, we evaluate the extensibility of \sys. We consider the case where the datasets owned by different \downers are biased, so it is desirable for the \mowner to incorporate the data from multiple \downers. 
In our experiment, we create five \downers, each with a dataset dominated by only two labels from the CIFAR10 dataset. For each training step, the \mowner trains the model with one of these \downers in turn. Finally \mowner tests the trained model with a test set with all ten labels. We evaluate with the two models used for CIFAR10. The results are shown in Figure~\ref{fig:multiple-does-acc}. Clearly, in the case where the datasets owned by different \downers are highly heterogeneous, extending the training with more \downers is necessary to achieve high test accuracies. 

\begin{figure}
    \centering
    \includegraphics[width=.85\linewidth]{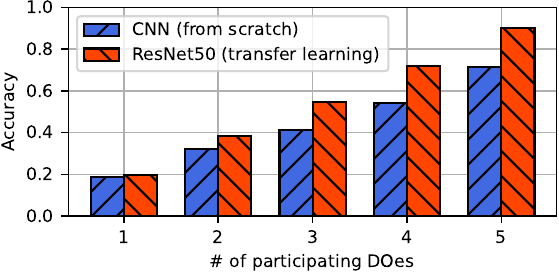}
    \caption{Test accuracies of the models trained with different numbers of heterogeneous \downers.}
    \label{fig:multiple-does-acc}
\end{figure}

\subsubsection{Impact of the DP Noises}\label{subsec:evaluation:dp-impact-on-accuracy}
In \S~\ref{subsec:design:linear:backward}, \sys designs a secret-shared version of DPSGD algorithm to allow the \downers to add DP noises to the weight updates. In this part, we evaluate the impact of DP noises on model performance. We use the two models of CIFAR10 in this experiment. We first estimate the gradient bound $C$ used in our algorithm to be $8$, then we conduct the experiments with different noise levels $\sigma$, with a fixed batch size $B=64$. The models are trained for 10 epochs. The results are shown in Table~\ref{table:dp-acc}. We observe that the accuracy slightly drops ($<2$\%) given $\sigma\leq 10^{-2}$. 
In \S~\ref{subsec:evaluation:gradient-attacks}, we empirically demonstrate that $\sigma=10^{-2}$ is sufficient to defend against the reconstruction attacks on the gradients.


\begin{table}
    \centering
    \begin{tabular}{rr|rr}
        \hline
        \multirow[c]{2}{*}{$\sigma$} & \multirow[c]{2}{*}{$(\epsilon,\delta)$-DP} & \multicolumn{2}{c}{Accuracy} \\
        & & \multicolumn{1}{c}{CNN} & \multicolumn{1}{c}{ResNet50} \\ \hline
        0     & None              & 71.69\% & 90.02\% \\
        0.005 & $( 114, 10^{-5})$ & 71.20\% & 89.71\% \\ 
        0.010 & $(  57, 10^{-5})$ & 70.32\% & 89.29\% \\
        0.020 & $(28.5, 10^{-5})$ & 68.87\% & 88.47\% \\
        0.050 & $(11.4, 10^{-5})$ & 57.01\% & 85.17\% \\
        \hline
    \end{tabular}
    \caption{Test accuracy of CIFAR10-CNN and ResNet50 trained with DP-mechanism of different noise levels $\sigma$.}
    \label{table:dp-acc}
\end{table}

\subsection{Performance Breakdown}
We present the performance breakdown of \sys. 

\parab{Linear Protocols}
We inspect the time and communication overhead required to train a single linear layer in \sys. The overhead of truncation is excluded. The results are shown in Table~\ref{table:performance-single-layer}. For training time, the preprocessing technique achieves $3\times\sim 5\times$ online speedup over the basic design. The communication cost is also reduced by $90\%$ for fully connected layers and $2/3$ for 2d-convolutional layers.

\begin{table}[htb]
    \centering \small
    \subfigure[Fully connected layers]{
        \begin{tabular}{rr|rr|rr}
            \hline
            & & \multicolumn{2}{c|}{\sys} & \multicolumn{2}{c}{\sysOpt} \\
            $n_i$ & $n_o$ & Time & Comm. & Time & Comm. \\ \hline 
            256  &  100 &   6.8ms &  0.40MB &  1.2ms &  18KB \\
            512  &   10 &   2.6ms &  0.17MB &  0.9ms &  14KB \\
            2048 & 1001 & 256.1ms &  5.51MB & 30.3ms & 822KB \\
            \hline
        \end{tabular}
    } \\ 
    \subfigure[2d-convolutional layers]{
        \begin{tabular}{rrrr|rr|rr}
            \hline
            & & & & \multicolumn{2}{c|}{\sys} & \multicolumn{2}{c}{\sysOpt} \\
            $c_i$ & $c_o$ & $h, w$ & $s$ & Time & Comm. & Time & Comm. \\ \hline 
            64 & 64 & $16^2$ & 5 &  165ms &  3.61MB &  13ms & 0.75MB \\
            3  & 64 & $56^2$ & 3 &  244ms & 11.61MB & 100ms & 4.71MB \\
            64 & 3  & $56^2$ & 3 &  237ms & 11.51MB &  86ms & 4.80MB \\
            \hline
        \end{tabular}
    }
    \caption{Online training overheads of linear layers. Batch size is 64 and the results are  averaged across samples. $n_i$ and $n_o$ denote the input and output neurons of a fully connected layer. $c_i, c_o$ denote input and output channels of a convolutional layer. $h,w$ is input image size and $s$ is the kernel size.}
    \label{table:performance-single-layer}
\end{table}

\parab{Acceleration of Hardware Implementation}
We further zoom into the performance of individual HE operators. In particular, we report the performance of our GPU-based BFV implementation and compare it with the state-of-the-art CPU implementation of Microsoft SEAL library \cite{SEAL}. For a complete comparison, we also implement the SIMD encoding/decoding functionality used in SEAL. 
The results are listed in Table~\ref{table:bfv-gpu-performance}. On average, our hardware implementation achieves $10\times$ acceleration over CPU implementation. We notice that there are other GPU acceleration works for HE (\eg \cite{Wang2014Accelerating, Turkoglu2022AnAccelerated}), but experimentally evaluating them is difficult since they are close-sourced.

\begin{table}[htb]
    \centering \small
    \begin{tabular}{r|rrr}
        \hline
        \multirow{2}{*}{HE Operation} & GPU & CPU & \multirow{2}{*}{Speedup} \\         
        & (Ours) & (\cite{SEAL}) \\  \hline
        Encoding       &  38 &   97 &  2.5$\times$ \\
        Decoding       &  65 &  114 &  1.8$\times$ \\
        Encryption     & 155 & 2611 & 16.8$\times$ \\
        Decryption     &  94 &  653 &  6.9$\times$ \\
        Addition       &   2 &   33 & 16.5$\times$ \\
        Multiplication & 506 & 8362 & 16.5$\times$ \\
        Relinearization&  38 &  173 &  4.6$\times$ \\
        Plain Mult.    & 178 & 1130 &  6.3$\times$ \\
        Rotation       & 366 & 1532 &  4.2$\times$ \\
        \hline
    \end{tabular}
    \caption{Time overhead ($\mu$s) of various HE operations.}
    \label{table:bfv-gpu-performance}
\end{table}

\subsection{Efficiency Comparison with Prior Art}\label{subsec:evaluation:comparison-prior-art}
In this section, we compare the training efficiency of \sys with prior art of MPC.
As discussed in \S~\ref{subsec:intro:related-work}, machine learning protocols using MPC rely on the non-colluding assumption, unless the \mowner and \downers themselves participate as computing servers, which, unfortunately, has suffer from the extensibility problems for the 2/3/4-PC frameworks or scalability problem for general $n$-PC frameworks.
Nevertheless, we compare the efficiency of \sys with two MPC frameworks, QUOTIENT~\cite{Agrawal2019quotient} and Semi2k~\cite{Cramer2018SPDZ2k}. The former is a non-extensible 2-PC framework, and the latter is extensible to any number of parties but we instantiate it with only 2 parties for maximum efficiency. The results are listed in Table~\ref{table:compare-quotient}. QUOTIENT does not provide communication costs in their paper, so we only compare its throughput. On average, \sys without and with preprocessing optimization respectively achieves a speed-up of $13\times$ and $40\times$ over \cite{Agrawal2019quotient} and 2 orders of magnitude over \cite{Cramer2018SPDZ2k}. 

\begin{table}[t]
    \centering \footnotesize
    \begin{tabular}{c|rrrr|rrr}
        \hline
        & \multicolumn{4}{c|}{Throughput ($10^4$ img/h)} & \multicolumn{3}{c}{Comm. (MB/img)} \\
        Model & \cite{Agrawal2019quotient} & \cite{Cramer2018SPDZ2k} & $\mathsf{P}$ & $\mathsf{P}^+$ & \cite{Cramer2018SPDZ2k} & $\mathsf{P}$ & $\mathsf{P}^+$ \\ \hline
        $2\times128\mathsf{FC}$ & 0.7 & 0.11 & 9.7 & 29.3 &  552 & 1.7 & 0.2 \\
        $3\times128\mathsf{FC}$ & 0.6 & 0.10 & 8.1 & 18.9 &  658 & 2.2 & 0.3 \\
        $2\times512\mathsf{FC}$ & 0.2 & 0.03 & 2.6 & 13.2 & 3470 & 5.2 & 0.8 \\ \hline
    \end{tabular}
    \caption{Performance comparison with QUOTIENT~\cite{Agrawal2019quotient} and Semi2k~\cite{Cramer2018SPDZ2k} in the 2 party setting. The models are represented as $n\times m\mathsf{FC}$, as used by \cite{Agrawal2019quotient}. $\mathsf{P}$ represents \sys and $\mathsf{P}^+$ represents \sysOpt.} 
    \label{table:compare-quotient}
\end{table}

For extensibility, we compare \sys with \cite{Cramer2018SPDZ2k} in the setting with over 2 parties on the MLP model of MNIST. The results are displayed in Table~\ref{table:compare-semi2k-mpc}. As a general $n$-PC architecture without the non-concluding assumption, the overhead of \cite{Cramer2018SPDZ2k} grows significantly with more parties. 
In contrast, the overhead of \sys remains the same regardless of the number of parties, because the multiparty training procedure is decomposed into 2-PC in every training step, and switching \downers between steps is cost-free.

\begin{table}[t]
    \centering \footnotesize
    \begin{tabular}{c|rrr|rrr}
        \hline
        & \multicolumn{3}{c|}{Throughput ($10^3$ img/h)} & \multicolumn{3}{c}{Comm. (per img)} \\
        Model & \cite{Cramer2018SPDZ2k} & \sys & \sysOpt & \cite{Cramer2018SPDZ2k} & \sys & \sysOpt \\ \hline
        2 parties & 1.11 & 97 & 265 &      0.55GB & 1.7MB & 0.2MB \\
        3 parties & 0.61 & 97 & 265 &      2.58GB & 1.7MB & 0.2MB \\
        4 parties & 0.41 & 97 & 265 &      6.06GB & 1.7MB & 0.2MB \\
        5 parties & 0.07 & 97 & 265 &     57.69GB & 1.7MB & 0.2MB \\ \hline
    \end{tabular}
    \caption{Performance comparison with Semi2k~\cite{Cramer2018SPDZ2k} in multiple party setting. The model is MNIST-MLP.} 
    \label{table:compare-semi2k-mpc}
\end{table}

\subsection{\sys Against Attacks}\label{subsec:evaluation:gradient-attacks}


\noindent\textbf{Gradient Matching Attack~\cite{Zhao2020idlg}} tries to reconstruct the original input using the model updates. \sys proposes a secret-shared version of DPSGD to prevent this attack. In this segment, we evaluate the effectiveness of our method. 
We use the train-from-scratch CNN on CIFAR10 to conduct the evaluation. Note the attack \emph{requires a very small batch size}, so we set the batch size to one.
The reconstruction results are shown in Figure~\ref{fig:gradient-matching-attack}, together with the original input image data. It is clear that a noise level of $10^{-2}$ is sufficient to protect the original data. In practice, the training batch size is much larger (\eg 64), which could further reduce the requirements of noise levels. 


\begin{figure}
    \centering
    \includegraphics[width=0.90\linewidth]{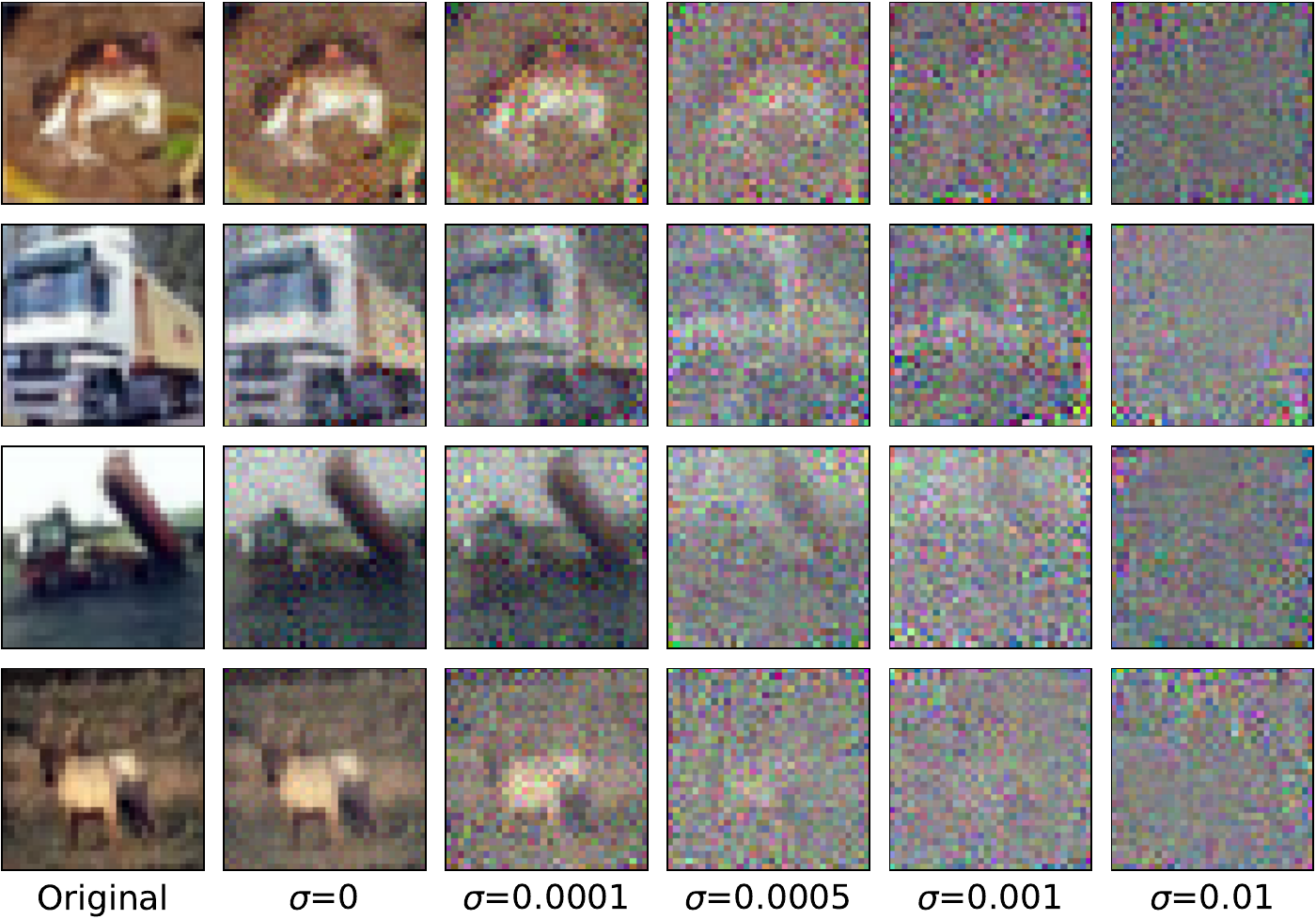}
    \caption{Gradient matching attack \cite{Zhao2020idlg} defended with different levels of noise}
    \label{fig:gradient-matching-attack}
\end{figure}

In addition, the attacker may try to reduce the noise by introducing a regularization term into the optimization goal \cite{Geiping2020inverting}. However, when the perturbations added to the gradients are large enough, such an attack would only produce smoothed but unidentifiable reconstructions. We defer the results in Appendix~\ref{appendix:inverting-gradients}.

\parab{Adaptive Attacks against the Preprocessing Design} 
We further consider an adaptive attack against the preprocessing design in \sys. In particular, a semi-honest \mowner may try to reconstruct the individual inputs from the linear combination of $\sum_{k\in[m']} a_k\vector{v}_k = \hat{\vector{v}}$ in Equation~(\ref{eq:v-linear-combination}) (for simplicity, we let $n_k=k$). 

Suppose \mowner tries to obtain \emph{one} specific element $v_i$ in $\vector{v}_i$ (\eg a specific pixel in an image). $v_i$ satisfy $\sum_{k\in[m']} a_k v_k = \hat{v}$. We launch the attack by the solving the problem stated in the proof sketch of Corollary~\ref{corollary:search-space} and
list the difficulty of the attack for several different sets of $m,f$ in Table~\ref{table:combination-breach-attack}, with a comparison to RSA modulus length $k$ offering equivalent bit security~\cite{NIST-SP800-57}.

We further empirically evaluate this attack on CIFAR10 dataset for a very small $m=2$ and precision $f=10$. \mowner presumes the pixel values of the normalized image are within $(-2, 2)$. It takes around 62.1 seconds for \mowner (single-thread, using the CPU stated in \S~\ref{subsec:evaluation:setup}) to obtain one pixel correctly. Thus, by proper setting of $m$ and $f$, \eg $m=8, f=25$, it would requires $9.5 \times 10^{55}$ seconds  (roughly $10^{48}$ years) to break one pixel. 
Symmetrically, 
it is equally difficult for a curious \downer to derive the model parameters.  

\begin{table}
    \centering
    \begin{tabular}{rr|rrr}
        \hline
        $m$ & $f$ & Search space & RSA-$k$ & Time to Break \\ \hline
        2 & 10 & 20 bits & $< 512$ & 62.1 seconds\\
        2 & 25 & 50 bits & $< 512$ & 2114 years\\
        4 & 25 & 100 bits & $\sim 2048$ & $2.38 \times 10^{18}$ years \\
        8 & 25 & 200 bits & $\sim 7680$ & $3.02 \times 10^{48}$ years \\ \hline
    \end{tabular}
    \caption{Hardness of the adaptive attack against the preprocessing optimization. RSA-$k$ means the RSA modulus bit length offering equivalent security guarantees. Time to break is evaluated or estimated using the CIFAR10 dataset.}
    \label{table:combination-breach-attack}
\end{table}

\section{Discussion and Future Work}

\subsection{Other Related Work}
We cover the related work that is not discussed in \S~\ref{subsec:intro:related-work}. 

\parab{Private Inference in Machine Learning} 
A series of art try to address the problem of private inference. 
In this setting, a client wishes to do inference on its private data on a third-party model deployed by a service provider. 
CryptoNets~\cite{Dowlin2016cryptonets} and Gazelle~\cite{Juvekar2018gazelle} tackles this problem using the HE primitives. 
\cite{Dowlin2016cryptonets} uses the SIMD technique of BFV within the batch size dimension, thus requiring a huge batch size of 8192 to reach maximum throughput. \cite{Juvekar2018gazelle} instead designs algorithms to exploit SIMD within one inference sample. It also combines garbled circuits (GC) for non-linear layers, while \cite{Dowlin2016cryptonets} changes non-linear layers to squaring function. On top of \cite{Juvekar2018gazelle}, Delphi~\cite{Mishra2020delphi} introduces a method to move all computationally heavy homomorphic operations to a preprocessing phase. However, this preprocessing must be executed once for every single sample on the fixed model weights. Thus, it is not applicable in training where the model weights are dynamic. 

Recent developments of CrypTFlow2~\cite{Rathee2020cryptflow2} and Cheetah~\cite{Huang2022cheetah} adapts oblivious transfer (OT) to substitute GC to improve efficiency. Further, Cheetah~\cite{Huang2022cheetah} proposes to apply the polynomial homomorphism instead of the SIMD technique in HE for linear computation. We adopt polynomial encoding as a primitive in \sys and adapt it for batched inputs in matrix multiplication and 2D convolution, similar to~\cite{Hao2022iron} and~\cite{liu2023llms}.

\parab{Hardware Acceleration in ML}
Hardware acceleration has been widely applied in the field of machine learning. 
Mainstream ML frameworks all support GPU or TPU acceleration, but only for plaintext ML. Extending hardware acceleration to privacy-preserving learning is a less charted area. Recently, \cite{Knott2021CrypTen, Tan2021CryptGPU} explores hardware acceleration for various MPC primitives, by moving various linear and non-linear operations onto GPU, under a setting of 3 or more servers. Some art further explores specialized FPGA or ASIC architecture. For instance, \cite{Zhou2022PPMLAC} accelerates the generation of multiplication triples using a trusted FPGA chip. 
There are also proposals to accelerate HE primitives. For instance, the widely used HE framework, Microsoft SEAL~\cite{SEAL}, supports multi-threading. \cite{Badawi2020PrivFT} implements a simplified version of CKKS cryptosystem (without relinearization or rescaling) on GPU. 


\subsection{Limitations and Explorations}


\parab{Parameter Selection} Private \emph{inference} framework of Cheetah~\cite{Huang2022cheetah} uses BFV with $t=2^{41}, q\approx 2^{109}, f=12$, with multiplicative depth of only 1. However, \sys need a much higher precision of $f=25$ bits to keep the gradients from diminishing into noises when training deep NNs. This also enforces a larger ciphertext modulus $q\approx 2^{180}$. Exploring a good balance between precision and HE parameters is an interesting direction. Recently, \cite{Rathee2021SIRNN} proposed to use mixed bit-widths in different NN operators to control the precision at a smaller granularity, which may be adapted to our framework.

\parab{Gradient Clipping} In \sys, to protect the privacy of model updates, \downer cannot directly clip the gradients as the original DPSGD algorithm~\cite{Abadi2016dpsgd}. Instead, \downer estimates an upper bound $C$ of gradients beforehand. Thus, the added DP noises may be slightly higher. A straightforward solution is to let \mowner compute the gradient norm in the encrypted domain and sends it to \downer, which, however, imposes non-trivial overhead. 
We leave further exploration to future work.  


\section{Conclusion}

In this work, we present \sys, a collaborative training framework that simultaneously achieves data privacy, model privacy, and extensibility of multiple data providers. \sys designs end-to-end training protocols by combining HE and MPC primitives for high efficiency and utilizes a preprocessing technique to offload HE operations into an offline phase. Meanwhile, we develop a highly-parallelized GPU version of the BFV HE scheme to support \sys. Evaluation results show that \sys achieves training accuracy nearly identical to plaintext training, while the training overhead is greatly reduced compared to prior art. Furthermore, we demonstrate that \sys is secure against both existing and adaptive attacks.

\section*{Acknowledgement}   
We thank the anonymous reviewers for their valuable feedback. The research is supported in part by the National Key R\&D Program of China under Grant 2022YFB2403900, NSFC under Grant 62132011 and Grant 61825204, and Beijing Outstanding Young Scientist Program under Grant BJJWZYJH01201910003011.

\bibliographystyle{IEEEtranS}
\bibliography{ref}

\appendix

\subsection{Polynomial encoding method}\label{appendix:polynomial-encoding}

We briefly introduce the polynomial encoding method to evaluate $\enc{\vector{y}} = \vector{W} \circ \enc{\vector{v}}$ in the BFV scheme, for $\circ$ as matrix multiplication and 2d-convolution (\ie $\conv(\vector{v};\vector{W})$). 

\parab{Matrix multiplication}
We use an improved version of the polynomial encoding method from~\cite{Huang2022cheetah}, as proposed by~\cite{Hao2022iron}. This improved version further takes the batch size dimension into account and reduces the communication costs for matrix multiplication.

Let weights be $\vector{W} \in \integers_t^{n_o \times n_i}$, and batched input be $\vector{v} \in \integers_t^{n_i \times B}$ (as $B$ column vectors). We assume $n_on_iB \leq N$, the polynomial degree of the BFV scheme. Larger matrices could be partitioned into smaller blocks, and the evaluation could be done accordingly. 

The input $\vector{v} = (v_{ij})_{i\in[n_i], j\in[B]}$ is encoded through $\pi_v$ as a polynomial
$$v = \pi_v(\vector{v}) = \sum_{j=0}^{n_i-1}\sum_{k=0}^{B-1} v_{jk} x^{kn_on_i + j}.$$
The weights $\vector{W} = (w_{ij})_{i\in[n_o], j\in[n_i]}$ are also encoded through $\pi_W$ as a polynomial (note how it is ``reversely'' encoded)
$$W = \pi_W(\vector{W}) = \sum_{i=0}^{n_o-1} \sum_{j=0}^{n_i-1} W_{ij}x^{in_i+n_i-1-j}.$$
$v$ is encrypted in the BFV scheme, and the evaluator could evaluate the polynomial product $W\cdot\enc{v}$ in the encrypted domain, obtaining an encrypted polynomial $\enc{y}$, where $y = Wv$. After decryption, $y = \sum_{i=0}^{n_on_iB-1}y_ix^i$ could be decoded through $\pi_y^{-1}$ to obtain the desired coefficients.
$$\vector{W}\cdot\vector{v} = \vector{y} = \pi_y^{-1}(y) = (y_{kn_on_i + in_i + n_i - 1})_{ik}$$
The correctness comes from the observation that the $(kn_on_i + in_i + n_i - 1)$-th term of $y$ is exactly
$$y_{kn_on_i + in_i + n_i - 1} = \sum_{j\in[n_i]}W_{ij}v_{jk}.$$

\parab{2D convolution}
We also improve the method of~\cite{Huang2022cheetah} for 2D convolution to adapt the batched input. Specifically, similar to the idea for matrix multiplication, we consider the output channel dimension and the batch size dimension in the polynomial encoding primitive.

Let weights $\vector{W}\in \integers_t^{c_o\times c_i\times s\times s}$, batched input $\vector{v}\in \integers_t^{B \times c_i\times h\times w}$. For simplicity, we assume $Bc_oc_ihw\leq N$. Larger inputs or weights could be partitioned to meet this requirement. 

The input $\vector{v} = (v_{b,c,i,j})_{b\in[B], c\in[c_i], i\in[h], j\in[w]}$ is encoded as a polynomial
\begin{align*}
    v = \pi_v(\vector{v}) = \sum_{b=0}^{B-1} \sum_{c=0}^{c_i-1} \sum_{i=0}^{h-1} \sum_{j=0}^{w-1} v_{b,c,i,j} x^{\mathsf{index}_v(b,c,i,j)} \\
    \mathrm{where~} \mathsf{index}_v(b,c,i,j) = bc_oc_ihw + chw + iw + j.
\end{align*}
The weights $\vector{W} = (W_{c',c,i,j})_{c'\in[c_o], c\in[c_i], i,j\in[h]}$ are also encoded reversely through $\pi_W$ as a polynomial
\begin{align*}
    W = \pi_W(\vector{W}) = \sum_{c'=0}^{c_o-1}\sum_{c=0}^{c_i-1}\sum_{i=0}^{h-1}\sum_{j=0}^{w-1}W_{c',c,i,j}x^{\mathsf{index}_W(c',c,i,j)} \\
    \mathrm{where~} \mathsf{index}_W(c',c,i,j) = O + c'c_ihw - chw - iw - j \\
    \mathrm{and~} O = (c_i-1)hw + (s-1)w + s-1.
\end{align*}
With this encoding, $\enc{y} = W\cdot \enc{v}$ has the desired result at the coefficients of specific terms. Precisely, if we denote $y = \sum_{i=0}^{Bc_oc_ihw-1}y_ix^i$, the decoding function is
\begin{align*}
    \vector{y} = \pi_y^{-1}(y) = (y_{\mathsf{index}_y(b,c',i,j)})_{b\in[B], c'\in[c_o], i\in[h], j\in[w]} \\
    \mathrm{where~} \mathsf{index}_y(b,c',i,j) = bc_oc_ihw + O + c'c_ihw + iw + j.
\end{align*}
Therefore, after decrypting $\enc{y}$, one could obtain the complete result of $\vector{y} = \conv(\vector{v}; \vector{W})$.

\subsection{Non-linear protocols}\label{appendix:nonlinear}

We briefly introduce the state-of-the-art ReLU, division (used in average pooling) and truncation protocols from prior art~\cite{Rathee2020cryptflow2, Huang2022cheetah} used in Pencil. These protocols are mostly based on oblivious transfers (OT) ~\cite{Yang2020Ferret}. In an $1$-out-of-$n$ oblivious transfer, there are two parties denoted as the Sender and the Receiver. The Sender provides $n$ messages while the Receiver takes $1$ of them, and Sender could not learn which one the Receiver takes while the Receiver could not learn any other messages except the chosen one.

\parab{ReLU} The ReLU activation function consists of a comparison protocol and a multiplexing (bit-injection) protocol.
$$\relu(x) = \drelu(x) \cdot x = \mathbf{1}\{x \geq 0\} \cdot x$$
$\mathbf{1}\{x>0\}$ could be obtained from the most significant bit of the secret shared $x = \share{x}_0 + \share{x}_1$, with
$$\begin{aligned}
& \mathbf{1}\{x \geq 0\} = \lnot \text{MSB}(x) \\
& = \text{MSB}(\share{x}_0) \oplus \text{MSB}(\share{x}_1) \oplus \mathbf{1}\{2^{\ell-1} - \share{x}_0 < \share{x}_1\}
\end{aligned}$$
The last term requires a secret comparison protocol where the two parties input $a = 2^{\ell-1} - \share{x}_0$ and $b = \share{x}_1$ respectively. The secret comparison protocol is implemented with OT and bit-and protocol~\cite{Rathee2020cryptflow2}. Specifically, the $\ell$-bit integer is decomposed into multiple blocks, each with $m$ bits. Then the two parties invokes 1-out-of-$2^m$ OT to obtain the greater-than and equality comparison results for each separate block. These comparison results of $\lceil \ell / m \rceil$ blocks are combined in a tree structure using a bit-and protocol (also implemented with OTs), according to the following observation:
$$\mathbf{1}\{a<b\} = \mathbf{1}\{a_1 < b_1\} \oplus (\mathbf{1}\{a_1 = b_1\} \land \mathbf{1}\{a_0<b_0\})$$
where $a=a_1||a_0, b = b_1||b_0$. This observation indicates that the comparison result of larger bit-length integers could be obtained from their smaller bit-length blocks. 

Finally, the multiplexing of ReLU is constructed with two OTs, where the two parties could select with their share of $\share{\drelu(x)}$, resulting in a canceling to zero when their selection bits are the same, and the original $x$ otherwise. This multiplexing protocol is detailed in Algorithm 6 of Appendix A in~\cite{Rathee2020cryptflow2}.

\parab{Division and Truncation} In average pooling, we need to use the division protocol where the shared integers are divided with a public divisor. In truncation, the shared integers are right-shifted $f$ bits. Our used neural networks only use average pooling of $2\times 2$ kernels. Therefore, the division used could also be considered a truncation of $2$ bits, and thus we only introduce the truncation protocol. Nevertheless, for the general case, the division protocol proposed by~\cite{Rathee2020cryptflow2} (Algorithm 9 in Appendix D) could be used.

For truncation, SecureML~\cite{Mohassel2017secureml} first proposed a protocol where each party directly performs local right-shift on its share. While lightweight, this protocol could introduce two kinds of errors: (1) a large error when the two shares' addition wraps over the ring; (2) a small error on the least significant bit. For accurate truncation, \cite{Rathee2020cryptflow2} instead invokes secret comparison protocols to compute the error terms and correct both the small and the large error by adding a correction term. Recently, \cite{Huang2022cheetah} observes that in private machine learning, as we are dealing with scaled decimals, the small 1-bit error would hardly affect the accuracies. Thus, \cite{Huang2022cheetah} proposes an approximate truncation protocol that only handles the possible large error, greatly reducing the overheads. 

\subsection{Security proofs}\label{appendix:security-proofs}

We give the proof for Theorem~\ref{theorem:basic-protocols}.

\begin{proof}
    We describe the simulators for data privacy and model privacy below, after which we provide a full hybrid argument that relies on the simulators.

    \parab{Simulator Protocol for Data Privacy}
    The simulator $\simulator$, when provided with the \mowner's view (including the model parameters $\vector{M}$, weight update $\grad{\vector{W}}, \grad{\vector{b}}$ and shares $\angleshare{\vector{X}}_0, \angleshare{\grad{\vector{Y}}}_0$ of each trainable linear layer in each training iteration), proceeds as follows:
    \begin{itemize}[nosep]
        \item[1.] $\simulator$ chooses an uniform random tape for the \mowner.
        \item[2.] $\simulator$ chooses keys $\pk, \sk$ for the HE scheme. 
        \item[3.] For every execution of Algorithm~\ref{alg:homomorphic-linear-evaluation}, instead of sending the encryption of its share $\enc{\angleshare{\vector{X}}_1}$, $\simulator$ simply sends encryption of a zero tensor $\enc{\vector{0}}$ (with the same shape as $\vector{X}$) to \mowner.
        \item[4.] For every execution of Algorithm~\ref{alg:weight-gradient-calculation}, similarly, $\simulator$ sends encryption of zeros instead of $\enc{\angleshare{\vector{X}}_1}, \enc{\angleshare{\grad{\vector{Y}}}_1}$ to \mowner. Moreover, for $\widetilde{\grad{\vector{W}}}$, it sends $\grad{\vector{W}} - \vector{s} + \angleshare{\grad{\vector{Y}}}_0 \circ \angleshare{\vector{X}}_0 + \vector{e}$.
        \item[5.] For non-linear layer evaluations, $\simulator$ uses randomized tensors as its share of the input. 
    \end{itemize}


    \parab{Simulator Protocol for Model Privacy}
    The simulator $\simulator$, when provided with the \downer's view (including the dataset $\mathcal{D}$ and the shares $\angleshare{\vector{X}}_1, \angleshare{\grad{\vector{Y}}}_1$ of each trainable layer in each training iteration), proceeds as follows:
    \begin{itemize}[nosep]
        \item[1.] $\simulator$ chooses an uniform random tape for the \downer.
        \item[2.] $\simulator$ receives $\pk$ of the HE scheme from the \downer.
        \item[3.] For every execution of Algorithm~\ref{alg:homomorphic-linear-evaluation}, the $\simulator$ receives $\enc{\angleshare{\vector{X}}_1}$ but sends $\enc{\angleshare{\vector{Y}}_1} = \enc{-\vector{s}}$ back, for some random tensor $\vector{s}$.
        \item[4.] For every execution of Algorithm~\ref{alg:weight-gradient-calculation}, similarly, the $\simulator$ sends only $\enc{-\vector{s}}$ back for decryption, instead of $\enc{\grad{\vector{W}}^{\mathsf{cross}} -\vector{s}}$.
        \item[5.] For non-linear layer evaluations, $\simulator$ uses randomized tensors as its share of the input. 
    \end{itemize}

    Now we show that the two simulated distribution is indistinguishable from the real-world distribution.

    \parab{Proof with a corrupted \mowner} 
    We show that the real world distribution is computationally indistinguishable from the simulated distribution via a hybrid argument. In the final simulated distribution, the simulator does not use the \downer's dataset as input, and so a corrupted \downer learns nothing in the real world.
    \begin{itemize}[leftmargin=*]
        \item $\hybrid_0$: This corresponds to the real world distribution where the \downer uses its training dataset $\mathcal{D}$.
        \item $\hybrid_1$: In this hybrid, we change the \downer's message in Algorithm~\ref{alg:homomorphic-linear-evaluation}. \downer sends the homomorphic encryption of a zero tensor $\enc{\vector{0}}$ instead of sending $\enc{\angleshare{\vector{X}}_1}$. It follows from the property of the homomorphic encryption (ciphertexts are computationally indistinguishable) that $\hybrid_1$ is indistinguishable from $\hybrid_0$.
        \item $\hybrid_2$: In this hybrid, similarly we change \downer's message in Algorithm~\ref{alg:weight-gradient-calculation}. \downer again uses encryption of zeros to substitute $\enc{\angleshare{\vector{X}}_1}$, $\enc{\grad{\vector{Y}}^C}$. $\hybrid_2$ is indistinguishable from $\hybrid_1$.
        \item $\hybrid_3$: In this hybrid, we further change \downer's behavior in Algorithm~\ref{alg:weight-gradient-calculation}: with the knowledge of \mowner's view, \downer sends $\grad{\vector{W}} - \vector{s} + \angleshare{\grad{\vector{Y}}}_0 \odot \angleshare{\vector{X}}_0 + \vector{e}$ as $\widetilde{\grad{\vector{W}}}$. Since this message exactly allows the \mowner to obtain $\grad{\vector{W}}$, thus $\hybrid_3$ is indistinguishable from $\hybrid_2$.
        \item $\hybrid_4$: In this hybrid, we replace \downer's inputs to the non-linear evaluation by random tensors. It follows from the security of these non-linear evaluation MPC protocols \cite{Huang2022cheetah, Rathee2020cryptflow2}, that $\hybrid_4$ is indistinguishable from $\hybrid_3$, completing the proof.
    \end{itemize}
    
    \parab{Proof with a corrupted \downer} 
    We show that the real-world distribution is computationally indistinguishable from the simulated distribution via a hybrid argument. In the final simulated distribution, the simulator does not use the \mowner's model weights $\vector{M}$ for training, and so a corrupted \downer learns nothing in the real world.
    \begin{itemize}[leftmargin=*]
        \item $\hybrid_0$: This corresponds to the real world distribution where the \mowner uses its model weights $\vector{M}$.
        \item $\hybrid_1$: In this hybrid, we change the \mowner's behavior in Algorithm~\ref{alg:homomorphic-linear-evaluation}. Instead of using the weights $\vector{W}$ to calculate $\vector{W}\circ \enc{\vector{X}} - \vector{s}$, \mowner simply samples $\vector{s}$ and sends back $\enc{-\vector{s}}$. Also, \mowner does not add bias term $\vector{b}$ to its share of output $\vector{Y}$. It follows from the security of homomorphic encryption that $\hybrid_1$ is computationally indistinguishable from $\hybrid_0$.
        \item $\hybrid_2$: In this hybrid, similarly we change the \mowner's behavior in Algorithm~\ref{alg:weight-gradient-calculation} by substituting the ciphertext message $\enc{\grad{\vector{W}}^{\mathsf{cross}} - \vector{s}}$ with simply $\enc{-\vector{s}}$. $\hybrid_2$ is computationally indistinguishable from $\hybrid_1$.
        \item $\hybrid_3$: In this hybrid, we replace \mowner's inputs to the non-linear evaluation by random tensors. It follows from the security of these non-linear evaluation MPC protocols \cite{Huang2022cheetah, Rathee2020cryptflow2}, that $\hybrid_3$ is indistinguishable from $\hybrid_2$, completing the proof.
    \end{itemize}
\end{proof}

\subsection{Neural Network Architecture}\label{appendix:nn-architecture}

\newtcolorbox{NNbox}
{
    arc=0pt,
    boxsep=1pt,left=2pt,right=2pt,top=2pt,bottom=2pt,
    colback=white,
}

\begin{figure}
    \centering
    \begin{NNbox}
        \small
        \begin{itemize}
            \item [(0)] Input shape $(1, 28, 28)$ flattened as $(784)$
            \item [(1)] Fully connected + ReLU: $128$ neurons $\to (128)$
            \item [(2)] Fully connected + ReLU: $128$ neurons $\to (128)$
            \item [(3)] Fully connected: $10$ neurons $\to (10)$
        \end{itemize}
    \end{NNbox}
    \caption{MLP for the MNIST task~\cite{Mohassel2018aby3, Patra2022blaze, Agrawal2019quotient}}
    \label{fig:nn-mnist-mlp}
\end{figure}
\begin{figure}
    \centering
    \begin{NNbox}
        \small
        \begin{itemize}
            \item [(0)] Input shape $(1, 28, 28)$
            \item [(1)] Conv2d + ReLU: $5$ channels, $5\times 5$ kernel, 2 stride, 2 margin padding $\to (5, 14, 14)$
            \item [(2)] Fully connected + ReLU: $100$ neurons $\to (100)$
            \item [(3)] Fully connected: $10$ neurons $\to (10)$
        \end{itemize}
    \end{NNbox}
    \caption{CNN for the MNIST task~\cite{Riazi2018Chameleon}}
    \label{fig:nn-mnist-cnn}
\end{figure}

\begin{figure}
    \centering
    \begin{NNbox}
    {\small
    \begin{itemize}
        \item [(0)] Input shape $(256, 64)$
        \item [(1)] Conv1d + ReLU: 128 channels, $5$ kernel $\to (128, 60)$
        \item [(2)] AvgPool1d: $2$ kernel $\to (128, 30)$
        \item [(3)] Conv1d + ReLU: 128 channels, $5\times 5$ kernel $\to (128, 26)$
        \item [(4)] AvgPool1d: $2$ kernel $\to (128, 13)$
        \item [(5)] Fully connected: $4$ neurons $\to (4)$
    \end{itemize}
    }
    \end{NNbox}
    \caption{TextCNN for the AGNews task~\cite{Zhang2015Agnews}}
    \label{fig:nn-agnews-textcnn}
\end{figure}

\begin{figure}
    \centering
    \begin{NNbox}
    {\small
    \begin{itemize}
        \item [(0)] Input shape $(3, 32, 32)$
        \item [(1)] Conv2d + ReLU: 64 channels, $5\times 5$ kernel, $2$ margin padding $\to (64, 32, 32)$
        \item [(2)] AvgPool2d: $2\times 2$ kernel $\to (64, 16, 16)$
        \item [(3)] Conv2d + ReLU: 64 channels, $5\times 5$ kernel, $2$ margin padding $\to (64, 16, 16)$
        \item [(4)] AvgPool2d: $2\times 2$ kernel $\to (64, 8, 8)$
        \item [(5)] Conv2d + ReLU: 64 channels, $3\times 3$ kernel, $1$ margin padding $\to (64, 8, 8)$
        \item [(6)] Conv2d + ReLU: 64 channels, $1\times 1$ kernel $\to (64, 8, 8)$
        \item [(7)] Conv2d + ReLU: 16 channels, $1\times 1$ kernel $\to (16, 8, 8)$
        \item [(8)] Fully connected: $10$ neurons $\to (10)$
    \end{itemize}
    }
    \end{NNbox}
    \caption{CNN for the CIFAR10 task~\cite{Tian2022sphinx}}
    \label{fig:nn-cifar10-cnn}
\end{figure}

\begin{figure}
    \centering
    \begin{NNbox}
    {\small
    \begin{itemize}
        \item [(0)] Input shape $(3, 224, 224)$
        \item [(1)] 
        In NN5: AlexNet~\cite{Krizhevsky2017AlexNet} pretrained feature extractor $\to (9216)$; \\
        In NN6: ResNet50~\cite{He2016ResNet} pretrained feature extractor $\to (2048)$
        \item [(2)] Fully connected + ReLU: $512$ neurons $\to (512)$
        \item [(3)] Fully connected + ReLU: $256$ neurons $\to (256)$
        \item [(4)] Fully connected: $10$ neurons $\to (10)$
    \end{itemize}
    }
    \end{NNbox}
    \caption{Transfer learning models for the CIFAR10 task}
    \label{fig:nn-cifar10-transfer}
\end{figure}

The neural networks used for the evaluation are listed in Figure~\ref{fig:nn-mnist-mlp},~\ref{fig:nn-mnist-cnn},~\ref{fig:nn-agnews-textcnn},~\ref{fig:nn-cifar10-cnn},~\ref{fig:nn-cifar10-transfer}. The neurons are the number of output dimensions for fully connected layers. The channels are the number of output channels for 2D convolution layers. For concision, the activation function ReLU is written directly after linear layers. The 1D convolution and 1D average pooling are implemented simply as a specialization of 2D variants.

\subsection{Inverting Gradients}\label{appendix:inverting-gradients}

In this section we evaluate the gradient inverting attack introduced by \cite{Geiping2020inverting}, which introduces a regularization term to remove the noises in the reconstruction results. We again use CNN and CIFAR10 for evaluating this attack. The results are shown in Figure~\ref{fig:gradient-inverting-attack}. The results confirm that, when $\sigma\geq 10^{-2}$ the reconstructed images are smoothed but completely unidentifiable.

\begin{figure}
    \centering
    \includegraphics[width=.90\linewidth]{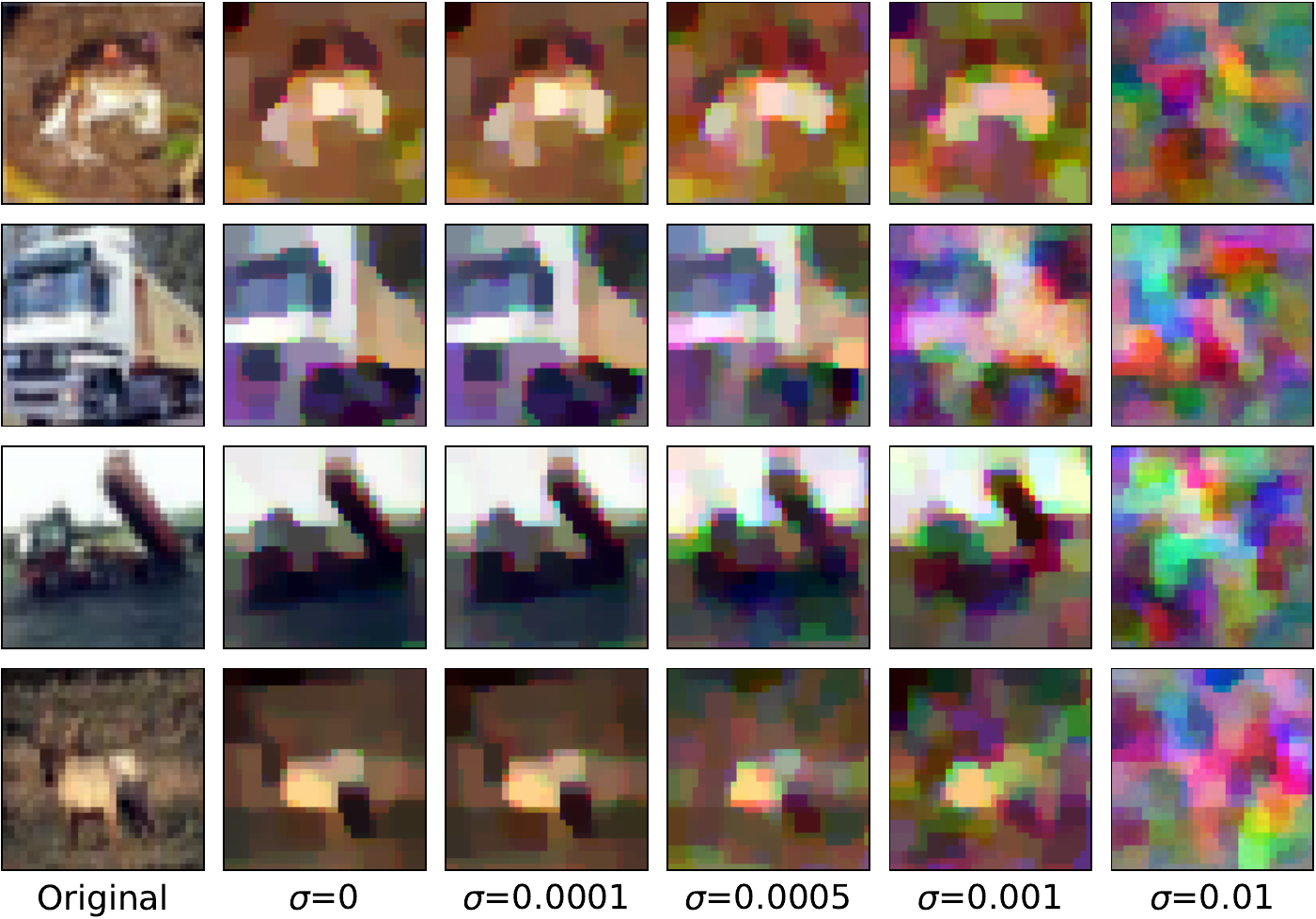}
    \caption{Gradient inverting attack~\cite{Geiping2020inverting} defended with different levels of noise}
    \label{fig:gradient-inverting-attack}
\end{figure}

\clearpage
\appendices
\section{Artifact Appendix}
\label{appendix:artifact}




\subsection{Description \& Requirements}




\subsubsection{How to access}
\begin{itemize}
    \item DOI: \href{https://doi.org/10.5281/zenodo.10140580}{https://doi.org/10.5281/zenodo.10140580}
    \item GitHub: \href{https://github.com/lightbulb128/Pencil}{https://github.com/lightbulb128/Pencil}
\end{itemize}


\subsubsection{Hardware dependencies}
A computer with CUDA supporting GPU devices. The code has been tested on a machine with a NVIDIA RTX A6000 GPU and CUDA 11.7.


\subsubsection{Software dependencies} 
A Linux machine, with python $\geq$ 3.8 and CUDA $\geq$ 11.7.
The code has been tested on Ubuntu 20.04.3 LTS.


\subsubsection{Benchmarks} 
MNIST, CIFAR10, AGNews datasets and pretrained ResNet50 and GPT-2 models are used. They are available from the \texttt{torchvision}, \texttt{torchtext} and \texttt{transformers} package.


\subsection{Artifact Installation \& Configuration}


The installation can be done via running \texttt{build\_dependencies.sh} provided in the repository. At a high level, it clones the three dependencies (\texttt{OpenCheetah}, \texttt{EzPC}, \texttt{seal-cuda}) and builds them. \texttt{OpenCheetah}, \texttt{EzPC} are 2-party secure computation libraries providing non-linearities evaluation, which are adapted to provide python interfaces, and \texttt{seal-cuda} is a homomorphic encryption library with GPU parallelization providing linearities related utilities.




\subsection{Major Claims}




\begin{itemize}
    \item (C1): Models trained privately with \sys has similar accuracy metrics as trained in plaintext, as shown by Table~\ref{table:training-accuracies} and Fig.~\ref{fig:acc-curve}. Reproduced by Experiment (E1).
    \item (C2): Training costs in \sys are acceptable as shown by Table~\ref{table:training-costs}. Reproduced by Experiment (E2).
    \item (C3): Including more DOes with heterogeneous data could improve model accuracy in \sys, as shown by Fig.~\ref{fig:multiple-does-acc}. Reproduced by Experiment (E3).
    \item (C4): One could integrate DP methods to provide stronger privacy guarantees over the model gradients, as shown by Table~\ref{table:dp-acc}. Reproduced by Experiment (E4).
\end{itemize}

\subsection{Evaluation}

The evaluation mainly focuses on two parts: model accuracy and training cost. In the provided repository, \texttt{gen\_scripts.sh} could be used to generate a series of evaluation scripts. These scripts are prefixed with \texttt{fullhe} or \texttt{prep}, corresponding to \sys or \sysOpt. Recall that \sys uses online homomorphic encryption while \sysOpt uses the preprocessing optimization which conveys all HE computations to the offline phase. Each script itself simulates the two participants DO and MO, which communicate through a localhost port. 

Each script represents a single experiment, either (1) to train a model and tests its accuracy, or (2) to measure the training cost of a single training step. The scripts will output logs in the corresponding \texttt{logs/} folder with the same name as the script, which will contain the accuracy and cost measurement results.

\subsubsection{Experiment (E0)} [\texttt{trivial.sh}] Train the simplest model on MNIST for 1 epoch. This is used for checking whether the installation is functional. Estimated time: 0.6 and 0.2 hours for \textsf{Pencil} and \textsf{Pencil}$^+$ respectively
\subsubsection{Experiment (E1)} [\texttt{train\_nn*.sh}] Train the model end-to-end and outputs the accuracy every 1/5 epochs. Estimated times are listed in Table \ref{tab:trainig-time}.
\subsubsection{Experiment (E2)} [\texttt{costs\_nn*.sh}] Train the model for one step and output the time and communication costs. For \textsf{Pencil}, cost evaluation time is expected to be below 1 minutes for each model. For \textsf{Pencil}$^+$, the time also includes the preprocessing phase time, listed in Table III in the main paper for each model.
\subsubsection{Experiment (E3)} [\texttt{dp*\_nn*.sh}] Train the model with different DP noise levels. Estimated times are listed in Table \ref{tab:trainig-time}.
\subsubsection{Experiment (E4)} [\texttt{hetero\_nn*.sh}] Train the model with 5 simulated DOes with different data distributions. Estimated times are listed in Table \ref{tab:trainig-time}.

The end-to-end online training time for each NN listed in Table \ref{tab:trainig-time} could also be calculated from the throughput listed in the paper.

\begin{table}
    \centering
    \begin{tabular}{cc|rr}
        \hline
        Task & Model & \textsf{Pencil} & \textsf{Pencil}$^+$ \\
        \hline
        MNIST & MLP & 6.2 & 2.3 \\
        MNIST & CNN & 7.8 & 4.4 \\
        AGNews & TextCNN & 162 & 78.9 \\
        CIFAR10 & CNN & 278 & 227 \\
        \hline
        CIFAR10 & AlexNet & 96 & 33.3 \\
        CIFAR10 & ResNet50 & 27.3 & 6.5 \\
        \hline
    \end{tabular}
    \vspace{0.5em}
    \caption{Online training time (hours) for each model.}
    \label{tab:trainig-time}
\end{table}









\subsection{Customization}
The provided code could be run with other arguments than listed in the generated scripts, including learning rate, optimizer type, gradient clipping bounds and noise levels, etc.. The detailed list and usage are provided in the \texttt{readme.md} file in the repository.



\end{document}